\documentclass[english, nolineno, numberwithinsect,cleveref,thm-restate]{socg-lipics-v2021}

\usepackage{amsmath}
\usepackage{amssymb}
\usepackage{xspace}
\usepackage{complexity}

\usepackage{todonotes}

\newif\iflong
\newif\ifshort
\longtrue
\iflong
\hideLIPIcs 
\nolinenumbers
\else
\shorttrue
\fi
\nolinenumbers

\title{Parameterised Partially-Predrawn Crossing Number}

\author{Thekla Hamm}{Algorithms and Complexity Group, TU Wien, Vienna, Austria}{thamm@ac.tuwien.ac.at}{https://orcid.org/0000-0002-4595-9982}{Supported by the Austrian Science Fund (projects P31336, Y1329, and \mbox{W1255-N23}).}

\author{Petr Hlin\v{e}n\'y}{Faculty of Informatics, Masaryk University, Brno, Czech Republic}{hlineny@fi.muni.cz}{https://orcid.org/0000-0003-2125-1514}
	{Supported by the {Czech Science Foundation}, project no.~{20-04567S}.}

\authorrunning{T.~Hamm and P.~Hlin\v{e}n\'y} 

\Copyright{Thekla Hamm and Petr Hlin\v{e}n\'y} 



\iflong
\relatedversion{A short version of the paper is to appear at SoCG~2022.}
\else
\relatedversion{A full version of the paper is available at \url{https://arxiv.org/abs/papernumber***********}.}
\fi

\EventEditors{Xavier Goaoc and Michael Kerber}
\EventNoEds{2}
\EventLongTitle{38th International Symposium on Computational Geometry (SoCG 2022)}
\EventShortTitle{SoCG 2022}
\EventAcronym{SoCG}
\EventYear{2022}
\EventDate{June 7--10, 2022}
\EventLocation{Berlin, Germany}
\EventLogo{socg-logo}
\SeriesVolume{224}
\ArticleNo{XX}  

\def\todo#1{\relax}
\def\apxmark{{\bf\hspace*{-1ex}*~}}

\newtheorem{case}{Case}

\def\ppdg{partially drawn\xspace} 
\def\ppd{partially predrawn\xspace} 
\def\PPD{Partially Predrawn\xspace} 
\def\ppdcr{\ppd crossing number\xspace}
\def\ppdcrc#1{\ppd #1-planar crossing number\xspace}
\def\PPDCR{\PPD Crossing Number\xspace}
\def\PPDCRc#1{\PPD #1-Planar Crossing Number\xspace}
\def\pdss{skeleton}\def\pdsg{predrawn \pdss\xspace} 
\def\crconf#1{#1-crossing conforming\xspace}  
\def\conf{\(k\)-crossing conforming\xspace}  
\def\Hep{\(H\)-edge path\xspace}  

\def\ca#1{\mathcal{#1}}
\def\cf#1{{\EuScript#1}}
\def\vect#1{\boldsymbol{#1}}
\def\crd{\operatorname{cr}} 
\def\crg{\operatorname{cr}} 
\def\crgpd{\operatorname{pd-cr}}
\def\prebox#1{\operatorname{\mbox{\slshape#1}}}

\begin{document}	
	\maketitle

\begin{abstract}
	Inspired by the increasingly popular research on extending partial graph drawings, we propose a new perspective on the traditional and arguably most important geometric graph parameter, the \emph{crossing number}.
	Specifically, we define the \emph{\ppdcr} to be the smallest number of crossings in any drawing of a graph, part of which is prescribed on the input (not counting the prescribed crossings).
	Our main result -- an \FPT-algorithm to compute the \ppdcr\ -- combines advanced ideas from research on the classical crossing number and so called \emph{partial planarity} in a very natural but intricate way.
	Not only do our techniques generalise the known \FPT-algorithm by Grohe for computing the standard crossing number, they also allow us to substantially improve a number of recent parameterised results for various drawing extension problems.
\end{abstract}
	
	\section{Introduction}
	Determining the crossing number, i.e.\ the smallest possible number of pairwise transverse intersections 
	(called {\em crossings}) of edges in any drawing, of a graph is among the most important problems in discrete computational geometry.
	As such its general computational complexity is well-researched:
	Probably most famously, it is known that graphs with crossing number \(0\), i.e.\ planar graphs, can be recognised in polynomial time~\cite{Wagner37,HopcroftT74,WeiKuanW99}.
	Generally, computing the crossing number of a graph is \NP-hard, even in very restricted settings~\cite{GareyJ83,Hlineny06,DBLP:journals/algorithmica/PelsmajerSS11,CabelloM13}, and also \APX-hard~\cite{Cabello13}.
	However there is a fixed-parameter algorithm for the problem, and even one that can compute a drawing of a graph with at most \(k\) crossings in time in \(\mathcal{O}(f(k)n)\) or decide that its crossing number is larger than \(k\)~\cite{Grohe04,KawarabayashiR07}.	
	
	More recently, so called \emph{graph drawing extension} problems have received increased attention.
	Instead of being given an entirely abstract graph as an input, here the input is a \emph{\ppdg graph} $\ca P=(G,\ca{H})$, meaning that a subgraph $H$ of the input graph $G$ is given with a fixed drawing $\ca H$ which must not be changed in the solution.
	This is motivated by immediate applications in network visualisation~\cite{MisueELS95}, as well as a more general line of research in which important computational problems are extended to the setting in which parts of the solution are prescribed which can lead to useful insights for dynamic or divide-and-conquer type algorithms and heuristics~\cite{CaselFMMS21,EibenGHK21}.
	In this context it is natural to define the \emph{\ppdcr} as the smallest number of pairwise crossings of edges in any drawing which coincides with (i.e., extends) the given fixed drawing of the \pdsg, minus the number of `unavoidable' crossings already contained in the fixed drawing of the \pdss.
	We name this problem \textsc{\PPDCR}.
	
	Of course, the problem of computing the \ppdcr is more general than the one of computing the classical crossing number (which is captured by the former by simply letting the \pdsg be empty), and thus the known hardness results for computing the classical crossing number carry over.
	To the best of our knowledge, the \ppdcr problem has so far not been explicitly studied in literature, although, there are papers which study \emph{partially embedded planarity}, i.e. the property of having \ppdcr \(0\), and variants thereof.
	In particular, similarly to ordinary planarity, \ppdg graphs extendable to planar drawings can be recognised in polynomial time~\cite{AngeliniDFJVKPR15}, 
	and in analogy to the Kuratowski theorem, there is also a neat list of forbidden ``\ppdg minors'' (Figure~\ref{fig:PEG-obstructions}) which characterise \ppdg graphs extendable to planar drawings~\cite{JelinekKR13}.
	
	If one allows a non-zero number of crossings, the only algorithmic results on extending \ppdg graphs with constrained crossings
	we are aware of are those for scenarios with a few edges or vertices outside of the \pdsg or/and with a small number of crossings for each edge.
	We give a brief list of these algorithmic results:
	\begin{itemize}
		\item An algorithm to determine the exact \ppdcr of a \ppdg graph in \FPT\ time parameterised by the number of edges which are not fixed by the \pdsg~\cite{DBLP:conf/compgeom/ChimaniH16} (the ``rigid'' case in the paper).
		\item An algorithm to determine whether there is a $1$-planar drawing (or more generally a drawing in which each edge outside of the \pdsg has at most \(c\) crossings) which coincides with the given partial drawing in \FPT\ time parameterised by (\(c\) and) the number of edges which are not fixed by the \pdsg~\cite{EibenGHKN20a,GanianHKPV21}.
		\item An algorithm to determine whether there is a $1$-planar drawing which coincides with the given partial drawing in \XP\ time parameterised by the vertex cover size of the edges which are not fixed by the \pdsg~\cite{EibenGHKN20b}.
		\item An algorithm to determine whether there is a simple drawing in which each edge outside of the \pdsg has at most \(c\) crossings which coincides with the given partial drawing in \FPT\ time parameterised by \(c\) and the number the edges which are not fixed by the \pdsg~\cite{GanianHKPV21}.
	\end{itemize}
	We remark that all these parameterised algorithms require the given \pdsg to be connected, and the last three algorithms are easily adapted to output drawings minimising the number of crossings under the requirement of the respective properties.

	\paragraph*{Contributions}
	The foundation of our main contribution is a fixed-parameter algorithm for an exact computation of the \ppdcr \(k\) of a given \ppdg graph.%
	\begin{restatable}{theorem}{thmmain} \label{thm:main}
		\textsc{\PPDCR} is in \FPT\ when parameterised by the solution value (i.e., by the number of crossings which are not predrawn).
	\end{restatable}
	We employ a technique similar to the approach showing fixed-parameter tractability of classical crossing number devised by Grohe~\cite{Grohe04}.
	This means we proceed in two phases:
	\begin{enumerate}[I.]
		\item We iteratively reduce the input \ppdg graph $\ca P$ until we cannot find a large flat grid in it, and so we bound its treewidth by a function of \(k\), or decide that the \ppdcr of $\ca P$ is larger than \(k\).
		Importantly, each reduction step is guaranteed to preserve the solution value (unless it is~$>k$).
		\item We devise an MSO\(_2\)-encoding for the property that any \ppdg graph has the \ppdcr at most~\(k\).
		The key idea is to encode the \pdsg of the input in a $3$-connected planar ``frame'' which is added to the input \ppdg graph.
		Using the bounded treewidth of the involved graph with the frame, we then apply Courcelle's theorem \cite{Courcelle90} in order to decide this property.
	\end{enumerate}
	Note that the second step is an interesting result in its own right:
	\begin{restatable}{lemma}{MSOphaseIIi} \label{lem:phaseIIi}
		For every $k\geq0$ there is an MSO\(_2\)-formula $\psi_k$ such that the following holds.
		Given a \ppdg graph $\ca P$, one can in polynomial time construct a graph $G'$ such that $\psi_k$ is true on $G'$ if and only if the \ppdcr of $\ca P$ is at most~\(k\).
		This claim holds also if some edges of $\ca P$ are marked as `uncrossable' and we compute the crossing number over such drawings extending $\ca P$ that do not have crossings on the `uncrossable' edges.
	\end{restatable}

	While our high-level approach is similar to Grohe's~\cite{Grohe04}, in each phase we are faced with some caveats, on which we elaborate in the respective sections, due to the fact that we must respect the given \pdsg and that we have to observe also the treewidth of the derived graph which encodes the \pdsg, i.e.~of $G'$ from Lemma~\ref{lem:phaseIIi}.

	In this regard, we also give a concrete example (see Proposition~\ref{pro:baddualdiam}) of a fundamentally different behaviour of the \ppdcr compared to the classical one
	(which can partly explain the difficulties we face in Theorem~\ref{thm:main}, compared to~\cite{Grohe04}).
	In a nutshell, we show that for fixed $k$ a \ppdg graph can have arbitrarily {many nested cycles}
	which are ``critical'' for having crossing number~$>k$.
	
	Based on the proof of \Cref{thm:main} we are also able to give an improved algorithm to determine whether there is a drawing in which each edge outside of the \pdsg has at most \(c\) crossings which coincides with the given partial drawing.
	Specifically we can show the following theorem, where the \emph{\ppdcrc{\(c\)}} of a \ppdg graph $\ca P$ is as the \ppdcr above while restricted to only drawings of~$\ca P$ in which each edge outside of the \pdsg has at most~\(c\)~crossings.
	\begin{restatable}{theorem}{cPlanar}\label{thm:cplanar-improve}
		\textsc{\PPDCRc{\(c\)}} is in \FPT\ when parameterised by the solution value (i.e., by the number of crossings which are not predrawn).
	\end{restatable}
	
	Compared to the algorithm given in \cite{EibenGHKN20a}, Theorem~\ref{thm:cplanar-improve} presents an additional improvement in two important aspects%
	\ifshort. Not only can our algorithm solve the \(c\)-planar drawing extension problem parameterised by the number of new crossings (a less restrictive parameter than the combination of \(c\) and \(|E(G) \setminus E(H)|\)), but we can also handle disconnected initial drawings.\fi\iflong:
	\begin{itemize}
		\item The \ppdcrc{\(c\)} is upper-bounded by the product of~\(c\) and the number of edges which are not fixed by the partial drawing.
		Conversely, the number of edges which are not fixed by the partial drawing can be arbitrary larger than the \ppdcrc{\(c\)}.
		This means our new algorithm is less restrictive in terms of the parameterisation.
		\item Our new algorithm does not require connectivity of the partial drawing.
	\end{itemize}\fi
	
	We also can combine our techniques with structural insights from \cite{GanianHKPV21} to drop the connectivity requirement on the input in the setting that we want to determine the \ppdcrc{\(c\)} restricted to simple drawings:
	\begin{restatable}{theorem}{Simple}\label{thm:Simple}
		Given a \ppdg graph, one can in \FPT\ time parameterised by \(c\) and the number of edges not contained in the \pdsg,
		decide the minimum number of crossings in a {\em simple} drawing which coincides with the given simple partial drawing and in which each edge outside of the \pdsg has at most \(c\) crossings.
	\end{restatable}

Further on, we defer proofs of the ~~\apxmark$\!\!$-marked statements to the Appendix.

	\section{Preliminaries}\label{sec:prelims}
	We use standard terminology for undirected simple graphs~\cite{Diestel12} and assume basic understanding of \emph{parameterised complexity}~\cite{CyganFKLMPPS15,DowneyF13}, and of \emph{Courcelle's theorem together with MSO logic}~\cite{ArnborgLS91,Courcelle90} and treewidth.
	We refer also to the \ifshort full preprint paper \fi\iflong Appendix \fi for additional background on these notions.
	Regarding embeddings and drawings of graphs we mostly follow~\cite{mtbook}.

	For $r \in \mathbb{N}$, we write $[r]$ as shorthand for the set $\{1,\ldots,r\}$.
	
	\subsection{Partial graph drawings}
	A \emph{drawing} $\mathcal{G}$ of a graph $G$ in the Euclidean plane $\mathbb{R}^2$ is a function that maps each vertex $v \in V(G)$ to a distinct point $\mathcal{G}(v) \in \mathbb R^2$ and each edge $e=uv \in E(G)$ to a simple open curve $\mathcal{G}(e) \subset \mathbb R^2$ with the ends $\mathcal{G}(u)$ and $\mathcal{G}(v)$.
	We require that $\ca G(e)$ is disjoint from $\ca G(w)$ for all~$w\in V(G)\setminus\{u,v\}$.
	In a slight abuse of notation we often identify a vertex $v$ with its image $\mathcal{G}(v)$ and an edge $e$ with~$\mathcal{G}(e)$.
	Throughout the paper we will moreover assume that:
	there are finitely many points which are in an intersection of two edges,
	no more than two edges intersect in any single point other than a vertex,
	and whenever two edges intersect in a point, they do so transversally (i.e., not tangentially).

	The intersection (a point) of two edges is called a \emph{crossing} of these edges.
	A drawing $\ca G$ is \emph{planar} (or a {\em plane graph}) if $\ca G$ has no crossings, and a graph is \emph{planar} if it has a planar drawing.
	The number of crossings in a drawing $\ca G$ is denoted by $\crd(\ca G)$.
	A drawing $\ca G$ is \emph{$c$-planar} (or a {\em$c$-plane graph}) if every edge in $\ca G$ contains at most $c$ crossings,
	and a graph is \emph{$c$-planar} if it has a $c$-planar drawing.
	The \emph{planarisation} \(\mathcal{G}^\times\) of a drawing \(\mathcal{G}\) of \(G\) is the plane graph obtained from \(\mathcal{G}\) by making each crossing point a new degree-$4$ vertex of~$\mathcal{G}^\times$.
	The inclusion-maximal connected subsets of the set-complement \(\mathbb{R}^2\setminus\ca G\) are called the \emph{faces} of~$\ca G$.
	For any drawing exactly one of these faces is infinite and referred to as the \emph{outer face}. 
	
	A \emph{partial drawing} of a graph \(G\) is a drawing of an arbitrary subgraph $H$ of \(G\).
	A \emph{\ppdg graph} \(\ca P=(G, \mathcal{H})\), with an implicit reference to~$H$, is a graph \(G\) together with a~partial drawing \(\mathcal{H}\) of \(H\subseteq G\), and then \(\mathcal{H}\) is called the {\em\pdsg} of \((G, \mathcal{H})\).
	\todo{We need to check and make this consistent in the end. I originally wanted a word for \(H\) and not \(\mathcal{H}\).}\todo{OK, we will see}
	We say that two drawings \(\mathcal{G}_1\) and \(\mathcal{G}_2\) of the same graph \(G\) are \emph{equivalent}
	if there is a homeomorphism of $\mathbb{R}^2$ onto itself taking \({\mathcal{G}_1}^\times\) onto \({\mathcal{G}_2}^\times\) \cite{mtbook}.
	For connected \({\mathcal{G}_1}^\times\) and \({\mathcal{G}_2}^\times\), this is the same as requiring equal rotation systems and the same outer face.
	However, for disconnected drawings, \cite{JelinekKR13} in addition to equal rotation systems and outer face it is neccessary to specify which faces of each connected component of \({\mathcal{G}_1}^\times\) contain which other connected components and in which orientation, and match this specification with \(\mathcal{G}_2^\times\) (see also \Cref{fig:disconnequiv}).

	In this setup, we also say that two {\em\ppdg graphs are isomorphic} if there exists an isomorphism which gives an equivalence of their \pdsg{s}.
	
	\begin{figure}
		\centering
		\includegraphics[page=1, scale=1.2]{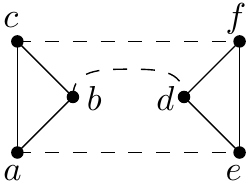}
		\hfil
		\includegraphics[page=2, scale=1.2]{disconnequiv.pdf}
		\caption{Two drawings of the same graph (solid lines) with the same rotation scheme. However both drawings are not equivalent. All dashed curves need to be mapped without crossing each other or any solid line by any homeomorphism from the left to the right the drawing. This is not possible.\label{fig:disconnequiv}}
	\end{figure}
	
	\subsection{Problem definitions}
	The \textsc{\PPDCR} problem takes as an input a partially drawn graph \((G, \mathcal{H})\) and an integer \(q\).
	The task is to decide whether there is a drawing \(\mathcal{G}\) of \(G\), the restriction of which to the \pdsg \(H\) is equivalent to \(\mathcal{H}\) (we can shortly say that $\ca G$ {\em extends}~$\ca H$), such that \(\mathcal{G}\) has at most \(q+\crd(\mathcal{H})\) crossings.
	The smallest value of the parameter \(q\) for which \((G, \mathcal{H})\) is a \texttt{yes}-instance of \textsc{\PPDCR} is called the \emph{\ppdcr} of \((G, \mathcal{H})\), denoted by $\crgpd(G, \mathcal{H})$.
	Note that $\crgpd(G,\emptyset)$ is the (called classical for distinction) \emph{crossing number} $\crg(G)$ of $G$.
	
	Likewise, the \textsc{\PPDCRc{$c$}} problem takes as an input a partially drawn graph \((G, \mathcal{H})\) and an integer \(q\).
	The task is to decide whether there is a drawing \(\mathcal{G}\) of \(G\) in which every edge in \(E(G) \setminus E(H)\) has at most \(c\) crossings and the restriction of which to \(H\) is equivalent to \(\mathcal{H}\), such that \(\mathcal{G}\) has altogether at most \(q+\crd(\mathcal{H})\) crossings.
	The smallest \(q\) (which may not be defined in general; a trivial example for which \(q\) is not defined is given by \(c=1\) and \(G\) not~$1$-planar) for which \((G, \mathcal{H})\) is a \texttt{yes}-instance of \textsc{\PPDCRc{$c$}} is called the \emph{\ppdcrc{$c$}} of \((G, \mathcal{H})\).

	\subsection{A parameterised algorithm for classical crossing number}	\label{sec:Grohes}
	We outline the high-level idea of Grohe's algorithm~\cite{Grohe04} to decide the classical crossing number of a graph in \FPT\ time and note some obstacles that we need to overcome.
	\ifshort Due to lack of space in the main paper, we leave the complete formal recapitulation together with some supplementary definitions for the full preprint paper.\fi
	\iflong See the Appendix for a complete formal recapitulation together with some supplementary definitions.\fi

	The algorithm proceeds in two phases.
	
	\paragraph*{Phase I -- Bounding Treewidth.}
	Consider a graph $G$ in which some edges are marked as `uncrossable', and the question of whether there is a drawing of \(G\) with at most \(k\) crossings in which no `uncrossable' edge is crossed for a fixed parameter~$k$.
	To improve readability, we shortly say that a drawing is {\em conforming} if no edge marked `uncrossable' is crossed in it.
	Grohe~\cite{Grohe04} showed that in polynomial time one can (i) confirm that the answer to this question is no, (ii) find a tree decomposition of \(G\) with width bounded in \(k\), or (iii) find a connected planar subgraph \(I \subseteq G\) where \(|V(I)| \geq 6\) together with a cycle \(C\) that is disjoint from \(V(I)\) and contains \(N(I)\) such that the following holds.
	If \(G'\) arises from \(G\) by contracting \(I\) to a vertex \(v_I\) and additionally marking all edges incident to \(v_I\) and all edges of \(C\) as `uncrossable', then any crossing-minimum conforming drawing of \(G\) arises from a crossing-minimum conforming drawing of \(G'\) by replacing \(v_I\) with a planar drawing of \(G[V(I) \cup V(C)]\) where the drawing of \(C\) is distorted to match that in the drawing of \(G'\) and \(I\) is drawn in an \(\varepsilon\)-neighbourhood of \(v_I\).
	Conversely, every crossing-minimum conforming drawing of \(G'\) arises from a crossing-minimum conforming drawing of \(G\) by contracting \(I\) and placing the resulting vertex on the drawing of some~vertex~in~\(I\).

	In the partially drawn setting we can however not simply contract a subgraph \(I\) without loosing information about its parts that are potentially fixed by the partial drawing of the instance.
	In particular, reinserting some unrestricted planar drawing of \(I\) can violate the partial drawing (see \Cref*{fig:flippedingrid}).

	\begin{figure}
		\centering
		\includegraphics[page=2]{flippedingrid.pdf}
		\hfil
		\includegraphics[page=1]{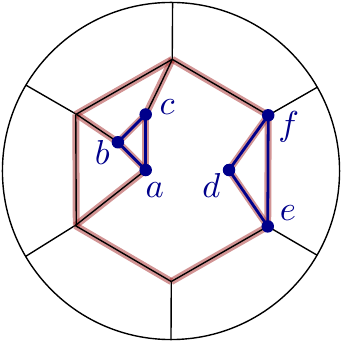}
		\caption{Example where predrawn parts (blue) make it impossible to simply insert a planar drawing of \(I\) (brown underlay).
		If the partial drawing is as on the left, \(I\) can be drawn planarly as depicted on the right but not while preserving equivalence of the partial drawing (cf.\ \Cref{fig:disconnequiv}). \label{fig:flippedingrid}}
	\end{figure}
	
	\paragraph*{Phase II -- MSO Encoding.}
	After having reduced \(G\) to a graph of treewidth bounded in the desired crossing number, one can apply Courcelle's theorem to decide whether $\crg(G)\leq k$ for any fixed~$k$.
	For that it is sufficient to encode in MSO\(_2\) logic the existence of at most \(k\) pairs of edges such that, after planarising a hypothetical crossing between the two edges of each pair, the resulting graph is planar.
	To express planarity, one simply excludes the existence of subdivisions of the two Kuratowski obstructions \(K_5\) and \(K_{3,3}\).
	The task of interpreting the planarisation of hypothetical crossings, ``guessed'' by existential quantifiers, is a more subtle one.
	In order to avoid heavy tools of finite model theory here, we can apply the following trick: instead of $G$, use the graph $G^{(k)}$ which subdivides $k$-times every edge of $G$, and ``guess'' $k$ pairs of the subdivision vertices which are pairwise identified to make the planarisation.

	This of course does not carry over easily to the partially drawn setting as the Kuratowski obstructions do not capture the \pdsg shape, i.e., there could be \ppdg graphs with high crossing number and not containing any $K_5$ or $K_{3,3}$ subdivisions.
	Here, instead, we will use the corresponding planarity obstructions for \ppdg graphs from \cite{JelinekKR13}, described next in Section~\ref{sec:PEGplanarity}.
	This brings two new complications to be resolved; namely that the list of obstructions is not finite, and that we have to encode the input drawing of the given \ppdg graph in an abstract way which can be ``read'' by an MSO$_2$-formula.

	\subsection{Characterising \ppd planarity}	\label{sec:PEGplanarity}

	We use the mentioned result of Jel\'{\i}nek, Kratochv\'{\i}l and Rutter \cite{JelinekKR13} characterising \ppd planarity, that is, the question of whether a given \ppdg graph $(G,\ca{H})$ admits a planar drawing which extends $\ca H$, by means of forbidding so-called PEG-minors.
	In this context we assume $\crd(\ca H)=0$.
	The forbidden obstructions are formed by one ``easy'' infinite family described separately (the {\em alternating chains}) and a list of $24$ specific \ppdg graphs shown in Figure~\ref{fig:PEG-obstructions}.
	However, since PEG-minors are not suitable for our application, we relax the characterisation of~\cite{JelinekKR13} to make a larger finite obstruction set
	and a simpler-to-handle containment relation (essentially a ``\ppdg topological~minor'').

	\begin{figure}
		\centering
		\includegraphics[width=0.75\hsize]{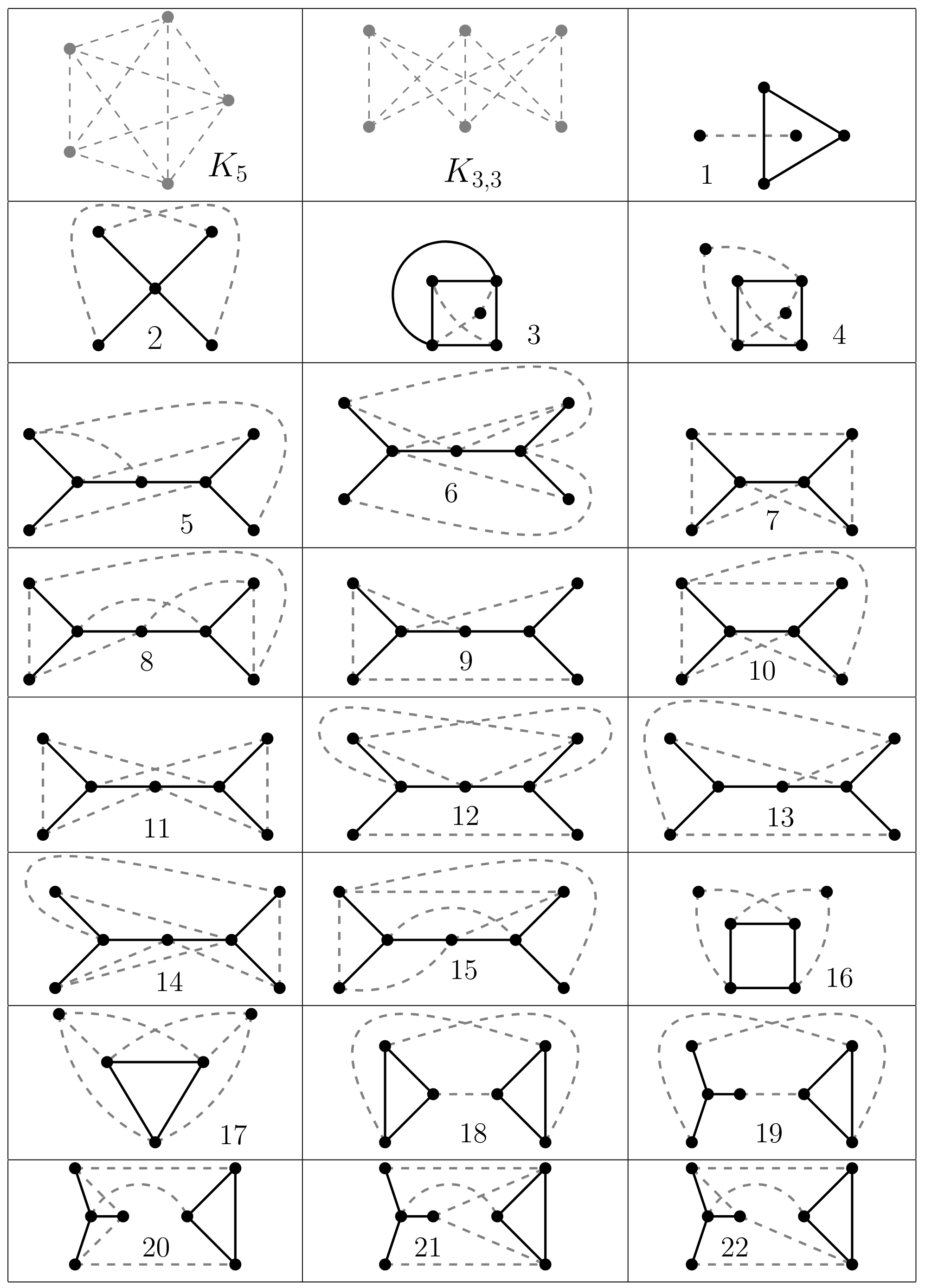}
		\caption{(A picture copied from \href{https://arxiv.org/pdf/1204.2915.pdf}{arXiv:1204.2915v1} with permission of the authors.)
		The list of $24$ \ppdg graphs~\cite{JelinekKR13} that are the obstructions (as PEG-minors) for \ppdg graphs which can be extended to planar drawings.
		The solid black edges and vertices form the \pdsg of the graphs, and dashed edges are the non-fixed ones.}
		\label{fig:PEG-obstructions}
	\end{figure}

	A {\em subdivision} of an edge in a \ppdg graph $(G,\ca{H})$ is the same subdivision in the graph $G$, which is correspondingly applied to $\ca{H}$ if the subdivided edge is from~$\ca{H}$.
	A \ppdg graph $(G_1,\ca{H}_1)$ is a {\em(\ppdg) subgraph} of $(G,\ca{H})$ if $G_1\subseteq G$, $H_1\subseteq H$ and the drawing $\ca{H}_1$ is equivalent to the restriction of $\ca{H}$ to~$H_1$.
	Note that in general one may have an edge of $G_1$ which is predrawn in $\ca H$ but not in~$\ca H_1$.

	\begin{restatable}[adapted from {\cite{JelinekKR13}}]{theorem}{PEGourversion}\label{thm:PEG-relaxed-obstructions}
		There is a finite family $\cf K$ of \ppdg graphs such that the following is true.
		A \ppdg graph $\ca P=(G,\ca{H})$ admits a planar drawing which extends $\ca H$ if and only if $\crd(\ca H)=0$ and the following hold:
		\begin{enumerate}[i.]
			\item there is no alternating chain in $\ca P$ (see the \ifshort preprint version \fi\iflong Appendix \fi for the full definition), and
			\label{it:noalter}
			\item no subdivision of a \ppdg graph from $\cf K$ is isomorphic to a \ppdg subgraph of $\ca P$.
			\label{it:nosubdiv}
		\end{enumerate}
	\end{restatable}
	
	\smallskip
	Briefly put, the family $\cf K$ from Theorem~\ref{thm:PEG-relaxed-obstructions} is composed of all graphs obtained from the obstructions $(G,\ca{H})$ in Figure~\ref{fig:PEG-obstructions} \cite{JelinekKR13} by possible iterative splittings (of vertices of degree $>3$ in~$G$) and possible releasing of certain edges from~$\ca H$.
	{The {\em splitting} of a vertex \(v\) is performed by partitioning the neighbourhood of \(v\) into two disjoint sets \(N_1\) and \(N_2\),
	and replacing \(v\) with two new adjacent vertices \(v_1\) and \(v_2\) such that the neighbourhood of $v_1$ is $N_1\cup\{v_2\}$ and the neighbourhood of $v_2$ is $N_2\cup\{v_1\}$.}
	The {\em release} of an edge $f\in E(H)$ from $\ca H$ is allowed if $f$ is a bridge, i.e.\ $f$ is not contained in any cycle of~$H$, and is performed as follows:
	If one end (resp., both ends) of $f$ is of degree $>2$ in $H$, subdivide $f$ once (twice), and denote by $f'$ the edge resulting from $f$ such that both ends of $f'$ are of degree $\leq2$ in~$H$.
	Then remove $f'$ only from $\ca H$ (but keep it in~$G$).
	\ifshort We leave the details for the full preprint paper.\fi
	\iflong We leave the details for the Appendix.\fi

	\section{Algorithm for \ppdcr}\label{sec:PPDCR}

	Note that, regarding the input \ppdg graph \((G,\mathcal{H})\), we may as well assume that $\ca H$ is a plane graph; otherwise, we replace $\ca H$ with its planarisation~$\ca H^\times$
	(and accordingly adjust~$G$, which formally means to move to the \ppdg graph \(\big((G-E(H))\cup \mathcal{H}^\times, \mathcal{H}^\times\big)\)).
	This is sound since neither do we care about the number of crossings prescribed by $\ca H$, nor do we have any restrictions on single edges in \(H\), and hence do not care to identify them.
	Thus, we will assume planar $\ca H$ throughout the rest of the section, unless we explicitly say otherwise.
	
	\subsection{Phase I -- Treewidth}
	\label{sec:phase1}
	
	To show that we can arrive at an input graph with small treewidth, we prove a statement analogous to Grohe's iterative contraction for the partially predrawn setting.
	Approaching this, however, it becomes quite clear that contracting a subgraph \(I\) must be treated much more delicately.
	The role of the cycle \(C\) in that case is that it could be treated as an interface to glue together two drawings -- any planar drawing of the contracted part and any drawing of \(G\) after contraction with at most \(k\) crossings in which no `uncrossable' edge is crossed.
	For actually gluing the parts together, the drawing of \(C\) might need to be `flipped' in either of these two drawings.
	This can create a problem in terms of being equivalent to \(\mathcal{H}\) on \(H\).
	Even if we ensure that each of the two drawings we would potentially like to glue together to a drawing of \(G\) are compatible with \(\mathcal{H}\) or the contraction of \(\mathcal{H}\), this compatibility is not invariant under flipping \(C\) (see e.g.\ \Cref{fig:incompatible}).
	
	\begin{figure}
		\begin{center}
		\includegraphics[scale=0.6]{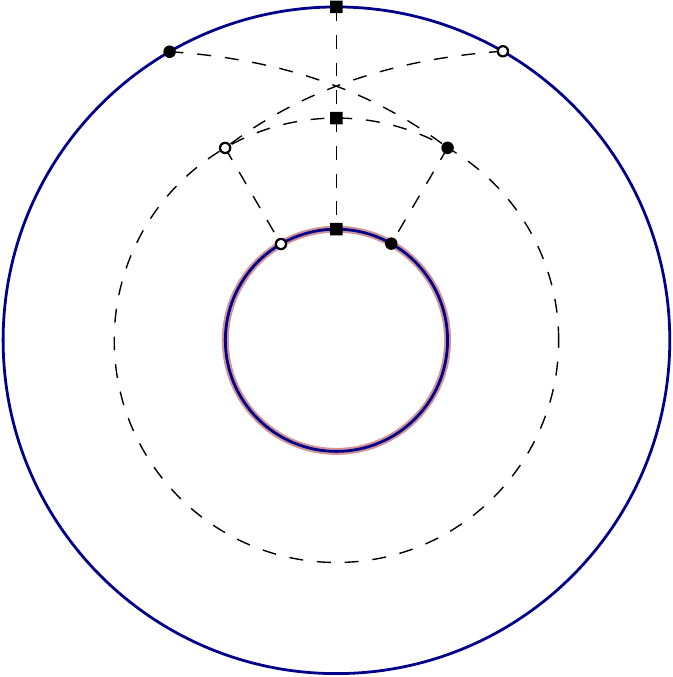} \hfill
		\includegraphics[scale=0.6,page=2]{compatibility.pdf} \hfill
		\includegraphics[scale=0.6,page=3]{compatibility.pdf}
		\end{center}
		\caption{Instance (left) cannot be drawn with \(0\) crossings (same node styles indicate adjacencies to vertices on dashed cycle),
			but subinstance (middle) consisting only of \(H\) (blue), and induced graph on \(I\) (brown underlay) and \(C\), as well as subinstance (right) in which \(I\) is contracted can.
			\label{fig:incompatible}}
	\end{figure}
	
	For this purpose we consider the notion of \emph{\(\mathcal{(H, I)}\)-flippability} for \(C\) and \(I\).
	Essentially, we say that \(C\) is \(\mathcal{(H, I)}\)-flippable in a graph \(D\), if the orientation of \(C\) with respect to \(I\) in a planar drawing of \(D\) that is equivalent to \(\mathcal{H}\) on \(H\) is not determined by \(\mathcal{H}\).
	Otherwise \(C\) is \(\mathcal{(H, I)}\)-unflippable in \(D\).
	A formal definition that makes use of the non-equivalence of drawing two disconnected triangles described in \Cref{fig:disconnequiv} is given in the
        \ifshort full preprint paper\fi\iflong Appendix\fi.
	Using this formal definition it can be decided in polynomial time whether a cycle is \(\mathcal{(H, I)}\)-flippable in a graph, or not.

	To facilitate readability, we say that for a \ppdg graph \((G, \mathcal{H})\) where some edges of \(G\) are marked as `uncrossable', the drawings of \(G\) that we want to consider, are \emph{\conf}.
	More formally, a \conf drawing is a drawing of \(G\) with at most \(k + \crd(\mathcal{H})\) crossings that is equivalent to \(\mathcal{H}\) on the \pdsg \(H\) and in which no `uncrossable' edge is crossed.
	The following key theorem is fully stated and proved in the \ifshort preprint paper\fi\iflong Appendix\fi.
	\begin{theorem}\label{thm:reducegrah}
		For all \(k \in \mathbb{N}\) there exists \(w \in \mathbb{N}\), such that given a \ppdg graph \((G, \mathcal{H})\) in which some edges are marked `uncrossable', in \FPT-time parameterised by \(k\) we~can
		\begin{enumerate}
			\item decide that there is no \conf drawing of \((G, \mathcal{H})\); or
			\item find a tree decomposition of \(G\) of width at most \(w\); or
			\item find an equivalent instance \((G', \mathcal{H}')\) with the property that \(|V(G')| < |V(G)|\).
		\end{enumerate}
	\end{theorem}
	\begin{proof}[Sketch of proof.]
		We start by applying the result by Grohe~\cite{Grohe04} for \(k\) with \(G\) as input.
		If the algorithm of~\cite{Grohe04} decides that the number of crossings in any drawing of \(G\) in which no `uncrossable' edge is crossed is more than \(k\) times, we can safely return that the same is true for any such drawing that is equivalent to \(\mathcal{H}\) on the \pdsg.
		Similarly, if the algorithm returns a tree decomposition of width at most \(w\), we can return that tree decomposition.
		
		In the last case, the algorithm finds a subgraph \(I \subseteq G\) and a cycle \(C\) in \(G\) as described in Subsection~\ref{sec:Grohes} for bounding treewidth.
		We distinguish whether there is a \crconf{\(0\)} drawing of \(\big(G[V(I) \cup V(C)] \cup H, \mathcal{H}\big)\), or not.
		Recall that, as we assume \(\mathcal{H}\) to be planarised, edges marked as `uncrossable' are irrelevant in this context because no edge should be crossed.
		Hence deciding whether there is a \crconf{\(0\)} drawing of \(\big(G[V(I) \cup V(C)] \cup H, \mathcal{H}\big)\) is equivalent to deciding whether \(\crgpd\big(G[V(I) \cup V(C)] \cup H, \mathcal{H}\big) = 0\).
		This can be decided in linear time using the result by Angelini et al.~\cite{AngeliniDFJVKPR15}.
		
		\begin{case}
			There is no \crconf{\(0\)} drawing of \(\big(G[V(I) \cup V(C)] \cup H, \mathcal{H}\big)\).
		\end{case}
		In this case we claim that there is no \conf drawing of \((G, \mathcal{H})\).
		Assume for a contradiction that there is such a drawing \(\mathcal{G}\).
		In particular this drawing has at most \(k\) crossings and no `uncrossable' edge is crossed in it.
		Hence, because of the choice of \(I\) and \(C\), no edge of \(G[V(I) \cup V(C)]\) is crossed in \(\mathcal{G}\).
		But as there are exactly \(\crd(\mathcal{H})\) crossings involving only edges of \(H\) in \(\mathcal{G}\), this means that the restriction of \(\mathcal{G}\) to \(G[V(I) \cup V(C)]\) is a \crconf{0} drawing of \(\big(G[V(I) \cup V(C)] \cup H, \mathcal{H}\big)\); a contradiction.
		
		\begin{case}
			There is a \crconf{\(0\)} drawing of \(\big(G[V(I) \cup V(C)] \cup H, \mathcal{H}\big)\).
		\end{case}
		This is the case in which we attempt to construct an equivalent instance with fewer vertices.
		Informally speaking, if we find an \(\mathcal{(H, I)}\)-flippable cycle \(C\), we will essentially be able to flip any planar drawing of the contracted subgraph to appropriately match the interface in a drawing of \(G\) after the contraction.
		Hence we can simply contract \(I\) in \(G\) and \(\mathcal{H}\).
		
		If we find a cycle that is \(\mathcal{(H, I)}\)-unflippable and the cycle remains unflippable after the contraction of the subgraph is performed, any planar drawing of the contracted subgraph automatically matches the interface in a drawing of \(G\) after contraction.
		Hence we can simply contract \(I\) in \(G\) and \(\mathcal{H}\).
		
		The last case is that the cycle we find is \(\mathcal{(H, I)}\)-unflippable but it seems to be flippable after the contraction of the subgraph is performed.
		In this case the orientation of the cycle is fixed in any planar drawing of the subgraph~$I$ for contraction, but both orientations of the cycle are possible after the contraction is performed.
		We must therefore appropriately force the orientation of \(C\) in the drawing after performing the contraction to match the one which is in fact forced before the contraction.
		We will do this by extending \(\mathcal{H}\) carefully.
	\end{proof}
	
	We can iteratively apply \Cref{thm:reducegrah} \(\mathcal{O}(|V(G)|)\) times to reduce our instance to a graph of small treewidth.
	Hence from now on we focus on the case that we are given a \ppdg graph \((G, \mathcal{H})\) and a tree decomposition of \(G\) whose width \(w\) is bounded in the inquired crossing number.
	
	This is already a crucial step towards the targeted application of Courcelle's theorem.
	However we still need to incorporate the information on the partial drawing \(\mathcal{H}\) into a graph structure of small treewidth.
	For this we will define a \emph{framing} of \((G,\mathcal{H})\).
	Note that even though we assume in this definition $\ca H$ to be planar, the definition also applies to the general case in which we first planarise $\ca H$ into $\ca H^\times$ and correspondigly adjust~$G$.
	
	\begin{definition}\label{def:framing}
		A \emph{framing} of a \ppdg graph \((G, \mathcal{H})\), where $\ca H$ is a plane graph, is an ordinary (abstract) graph~$F$ constructed as follows.
		See Figure~\ref{fig:framingex}.
		We start with the initial drawing \(\ca D:=\mathcal{H}\) and continue by the following steps in order:
		\begin{enumerate}
			\item While the graph of $\ca D$ is not connected, we iteratively add edges from \(G\) to $\ca D$ that can be inserted in a planar way and which connect two previously disconnected components.
			If this is no longer possible while the graph is still disconnected, let \(B\) be a face of $\ca D$ incident to more than one connected component.
			We pick a vertex \(v\) on \(B\) and connect \(v\) to an arbitrary vertex from each component incident to \(B\) which does not contain \(v\).
			We will call all edges added in this step the \emph{connector edges} (of the resulting framing).
			\item We replace each edge $f=uw$ of the drawing $\ca D$ from Step 1 (including the connector edges) by three internally disjoint paths of length~$3$ between $u$ and~$w$.
			We will call these three paths together the {\em framing triplet of~$f$}, and denote by $\ca D'$ the resulting drawing.
			\item Around each vertex $v\in V(\mathcal{H}^\times)$ in the drawing $\ca D'$ from Step 2, we add a cycle on the neighbours of $v$ in $\ca D'$ in the cyclic order given by $\ca D'$.
			We will call these cycles the \emph{framing cycles}, and all edges of the resulting planar drawing $\ca D''$ the {\em frame edges}.
			\item Finally, we set $F:=D''\cup G$ where $D''$ is the underlying graph of $\ca D''$ from Step 3.
		\end{enumerate}
	\end{definition}

	We remark that Step 1 of the construction of a framing $F$ of \((G, \mathcal{H})\) is not deterministic, and hence a \ppdg graph can admit multiple framings.
	Note also that possible connector edges introduced in Step 1 are no longer present in resulting $F$ (only their vertices and derived frame triplets are present).
	Moreover, the most important aspect of Definition~\ref{def:framing} is that the frame ($\ca D''$) defined after Step 3 is a $3$-connected planar graph which hence combinatorially captures the drawing \(\mathcal{H}\) within the framing~$F$.

	\begin{figure}[t]
	\begin{center}
	\begin{tikzpicture}[scale=0.7]
	\gdef\uugraph{%
		\tikzstyle{every node}=[draw, shape=circle, minimum size=2pt,inner sep=1.4pt, fill=black]
		\tikzstyle{every path}=[color=black]
		\node at (-2,0) (a) {}; \node at (2,0) (b) {}; \node at (0,2) (c) {};
		\node at (-2,10) (d) {}; \node at (2,10) (e) {}; \node at (0,5) (f) {};
		\node at (0,8) (g) {};
		\tikzstyle{every path}=[color=blue,thick]
		\draw (a) -- (b) -- (c) -- (a);
		\draw (d) -- (g) -- (e) (g) -- (f);
		\tikzstyle{every path}=[color=green!50!black]
		\draw (d) -- (a) to[bend left=18] (f) (f) to[bend left=18] (b) (b) -- (e);
		\draw (c) to[bend left=55] (g);
		\node at (2,6) (h) {};
	}\uugraph
		\draw (c) -- (f);
		\path[use as bounding box] (-2,-1) rectangle (2,11);
	\end{tikzpicture}
	\qquad\qquad\qquad
	\begin{tikzpicture}[scale=0.7]
		\uugraph
		\draw (c) to[bend right=30] (f);
		\draw[dotted] (c) -- (f);
		\tikzstyle{every path}=[color=red!70!black]
		\draw (a) to[bend left=12] (c) (a) to[bend left=24] (c) (a) to[bend right=12] (c);
		\draw (c) to[bend left=12] (b) (c) to[bend left=24] (b) (c) to[bend right=12] (b);
		\draw (b) to[bend left=12] (a) (b) to[bend left=24] (a) (b) to[bend right=12] (a);
		\draw (c) to[bend left=12] (f) (c) to[bend left=24] (f) (c) to[bend right=12] (f);
		\draw (d) to[bend left=12] (g) (d) to[bend right=24] (g) (d) to[bend right=12] (g);
		\draw (e) to[bend left=12] (g) (e) to[bend left=24] (g) (e) to[bend right=12] (g);
		\draw (f) to[bend left=12] (g) (f) to[bend left=24] (g) (f) to[bend right=12] (g);
		\draw (a) circle(1); \draw (b) circle(1); \draw (c) circle(1);
		\draw (d) circle(1); \draw (e) circle(1); \draw (f) circle(1);
		\draw (g) circle(1);
		\tikzstyle{every node}=[draw, shape=circle, minimum size=1pt, inner sep=0.75pt, fill=red]
		\node at (-1.01,0.18) {}; \node at (-1.01,-0.18) {}; \node at (-1.06,-0.36) {};
		\node at (1.01,0.18) {}; \node at (1.01,-0.18) {}; \node at (1.06,-0.36) {};
		\node at (-1.18,0.58) {}; \node at (-1.4,0.81) {}; \node at (-1.55,0.9) {};
		\node at (1.18,0.58) {}; \node at (1.4,0.81) {}; \node at (1.55,0.9) {};

		\node at (-0.17,2.99) {}; \node at (0.17,2.99) {}; \node at (-0.34,2.96) {};
		\node at (-0.57,1.19) {}; \node at (-0.8,1.41) {}; \node at (-0.89,1.55) {};
		\node at (0.57,1.19) {}; \node at (0.8,1.41) {}; \node at (0.89,1.55) {};
		\node at (-0.17,5.99) {}; \node at (0.17,5.99) {}; \node at (-0.34,5.96) {};
		\node at (-0.17,4.01) {}; \node at (0.17,4.01) {}; \node at (-0.34,4.04) {};

		\node at (-0.17,7.01) {}; \node at (0.17,7.01) {}; \node at (-0.34,7.04) {};
		\node at (-0.57,8.81) {}; \node at (-0.8,8.59) {}; \node at (-0.89,8.45) {};
		\node at (-0.57,8.81) {}; \node at (-0.8,8.59) {}; \node at (-0.89,8.45) {};
		\node at (0.57,8.81) {}; \node at (0.8,8.59) {}; \node at (0.89,8.45) {};
		\node at (1.57,9.1) {}; \node at (1.42,9.19) {}; \node at (1.19,9.41) {};
		\node at (-1.57,9.1) {}; \node at (-1.42,9.19) {}; \node at (-1.19,9.41) {};
	\end{tikzpicture}
	\end{center}
		\caption{(Definition~\ref{def:framing})
		A \emph{framing} of a \ppdg graph \((G, \mathcal{H})\): the graph is on the left, such that the \pdsg $\ca H$ is drawn with thick blue edges and the remaining edges of $E(G)\setminus E(H)$ are in green.
		The framing of \((G, \mathcal{H})\) on the right has the frame edges drawn in red; for every edge of $\ca H$ and for the chosen one connector edge between the two components of $H$, we get a framing triplet, and for every vertex of $\ca H$ a framing cycle.}
		\label{fig:framingex}
	\end{figure}
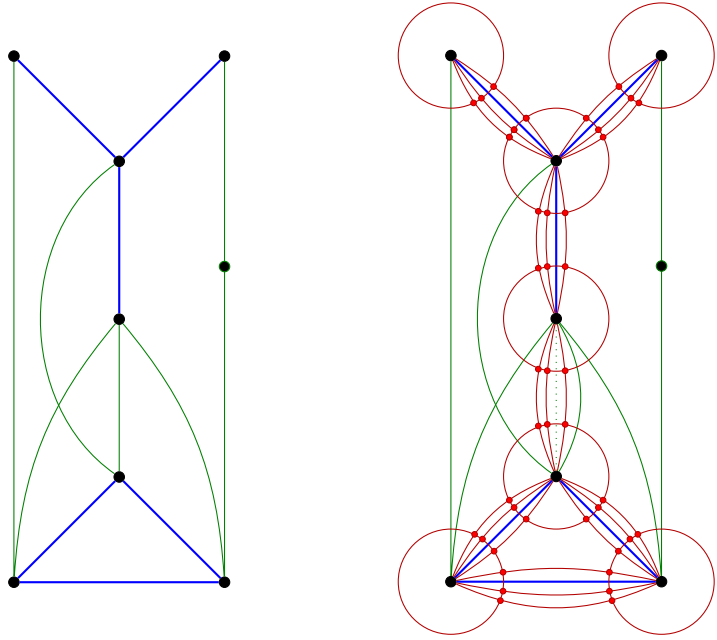

	As the last step in preparation for applying Courcelle's theorem we need to show that the framing construction does not considerably increase the treewidth:
	
	\begin{restatable}{lemma}{twbound}
		\label{lem:tw}\apxmark
		Let \(F\) be a framing of a \ppdg graph \((G, \mathcal{H})\), and $G^o=(G - E(H)) \cup \mathcal{H}^\times$.
		Then \(\operatorname{tw}(F) \in \mathcal{O}(16^{k + 1}\operatorname{tw}(G^o)/\log(\operatorname{tw}(G^o)))\), where \(k = \crgpd(G, \mathcal{H})\).
	\end{restatable}

	\subsection{Phase II -- MSO\(_2\)-encoding}\label{sec:MSOenc}
	\noindent
	Our aim now is to prove key Lemma~\ref{lem:phaseIIi}. In closer detail, we are first going to show:
	\begin{restatable}{lemma}{MSOforobstruction}\label{lem:MSOforobstruction}\apxmark
		Let $\ca P_1=(G_1,\ca{H}_1)$ be a \ppdg graph where $\ca{H}_1$ is plane. There exists an MSO$_2$-formula~$\sigma$, depending on $\ca P_1$, such that the following is true:
		\begin{itemize}\item
		For any \ppdg graph~$\ca P_2=(G_2,\ca{H}_2)$ with plane $\ca{H}_2$ and any framing $\bar G_2$ of $\ca P_2$ we have that $\bar G_2\models\sigma$,
		if and only if some subdivision of $\ca P_1$ is a \ppdg~subgraph~of~$\ca P_2$.
		\end{itemize}
	\end{restatable}

	\smallskip
	To combinatorially characterise the \ppdg subgraph containment, we use Definition~\ref{def:framing} and the following concept of a ``framing-aware'' minor.
	Considering framings $\bar G_1$ of $(G_1,\ca{H}_1)$ and $\bar G_2$ of $(G_2,\ca{H}_2)$, we say that $\bar G_1$ is a {\em framing topological minor} of $\bar G_2$ if there is a topological-minor embedding of $\bar G_1$ into $\bar G_2$ which additionally satisfies
	\begin{itemize}
		\item every edge of $G_1$ (resp., of $H_1$) is mapped into a path of $G_2$ (resp., of $H_2$),
		\item every framing cycle in $\bar G_1$ is mapped into a corresponding framing cycle in $\bar G_2$,
		\item whenever an edge $f\in E(H_1)$ is mapped into a path $P_f\subseteq H_2$, the framing triplet of $f$ in $\bar G_1$ is embedded (as three internally-disjoint paths) in the union of the framing cycles and triplets of the internal vertices and edges of $P_f$ in~$\bar G_2$, and
		\item the analogous condition (as the previous point) applies also to framing triplets of the connector edges of $\bar G_1$, which are embedded in~$\bar G_2$.
	\end{itemize}
	See Figure~\ref{fig:framingminor} for a natural illustration of this concept.

	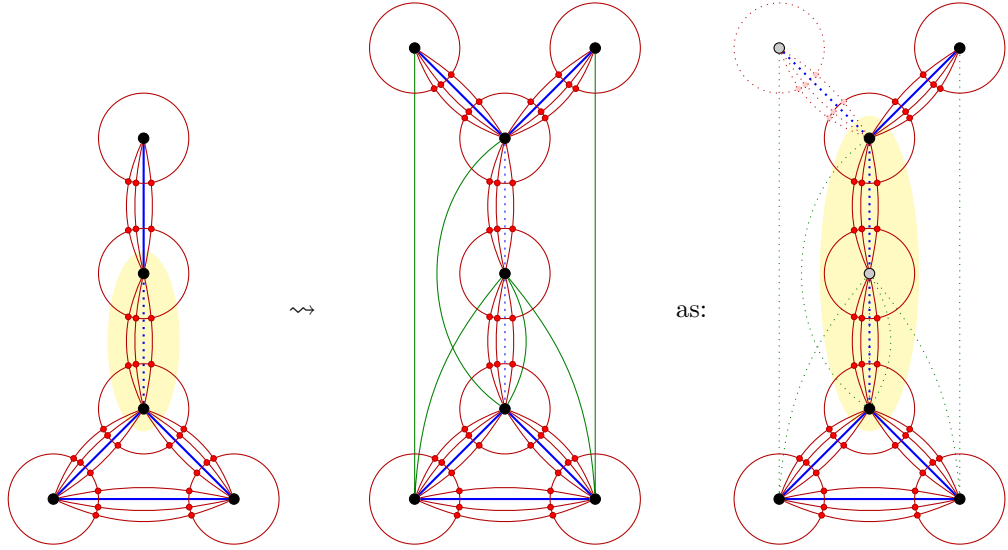
\begin{figure}[t]
	\begin{center}
	\begin{tikzpicture}[scale=0.6]
		\draw[draw=none,fill=white!70!yellow] (0,3.5) ellipse (0.8 and 2.0);
		\tikzstyle{every node}=[draw, shape=circle, minimum size=2pt,inner sep=1.4pt, fill=black]
		\tikzstyle{every path}=[color=black]
		\node at (-2,0) (a) {}; \node at (2,0) (b) {}; \node at (0,2) (c) {};
		\node at (0,5) (f) {}; \node at (0,8) (g) {};
		\tikzstyle{every path}=[color=blue,thick]
		\draw (a) -- (b) -- (c) -- (a);
		\draw (g) -- (f);
		\draw[dotted] (c) -- (f);
		\tikzstyle{every path}=[color=red!70!black]
		\draw (a) to[bend left=12] (c) (a) to[bend left=24] (c) (a) to[bend right=12] (c);
		\draw (c) to[bend left=12] (b) (c) to[bend left=24] (b) (c) to[bend right=12] (b);
		\draw (b) to[bend left=12] (a) (b) to[bend left=24] (a) (b) to[bend right=12] (a);
		\draw (c) to[bend left=12] (f) (c) to[bend left=24] (f) (c) to[bend right=12] (f);
		\draw (f) to[bend left=12] (g) (f) to[bend left=24] (g) (f) to[bend right=12] (g);
		\draw (a) circle(1); \draw (b) circle(1); \draw (c) circle(1);
		\draw (f) circle(1); \draw (g) circle(1);
		\tikzstyle{every node}=[draw, shape=circle, minimum size=1pt, inner sep=0.75pt, fill=red]
		\node at (-1.01,0.18) {}; \node at (-1.01,-0.18) {}; \node at (-1.06,-0.36) {};
		\node at (1.01,0.18) {}; \node at (1.01,-0.18) {}; \node at (1.06,-0.36) {};
		\node at (-1.18,0.58) {}; \node at (-1.4,0.81) {}; \node at (-1.55,0.9) {};
		\node at (1.18,0.58) {}; \node at (1.4,0.81) {}; \node at (1.55,0.9) {};

		\node at (-0.17,2.99) {}; \node at (0.17,2.99) {}; \node at (-0.34,2.96) {};
		\node at (-0.57,1.19) {}; \node at (-0.8,1.41) {}; \node at (-0.89,1.55) {};
		\node at (0.57,1.19) {}; \node at (0.8,1.41) {}; \node at (0.89,1.55) {};
		\node at (-0.17,5.99) {}; \node at (0.17,5.99) {}; \node at (-0.34,5.96) {};
		\node at (-0.17,4.01) {}; \node at (0.17,4.01) {}; \node at (-0.34,4.04) {};
		\node at (-0.17,7.01) {}; \node at (0.17,7.01) {}; \node at (-0.34,7.04) {};
	\end{tikzpicture}
	\raise20ex\hbox{\large$\leadsto$}\qquad
	\begin{tikzpicture}[scale=0.6]
		\tikzstyle{every node}=[draw, shape=circle, minimum size=2pt,inner sep=1.4pt, fill=black]
		\tikzstyle{every path}=[color=black]
		\node at (-2,0) (a) {}; \node at (2,0) (b) {}; \node at (0,2) (c) {};
		\node at (-2,10) (d) {}; \node at (2,10) (e) {}; \node at (0,5) (f) {};
		\node at (0,8) (g) {};
		\tikzstyle{every path}=[color=blue,thick]
		\draw (a) -- (b) -- (c) -- (a);
		\draw (d) -- (g) -- (e);
		\draw[dotted,thin] (g) -- (f) -- (c);
		\tikzstyle{every path}=[color=green!50!black]
		\draw (d) -- (a) to[bend left=18] (f) (f) to[bend left=18] (b) (b) -- (e);
		\draw (c) to[bend left=55] (g);
		\draw (c) to[bend right=30] (f);
		\tikzstyle{every path}=[color=red!70!black]
		\draw (a) to[bend left=12] (c) (a) to[bend left=24] (c) (a) to[bend right=12] (c);
		\draw (c) to[bend left=12] (b) (c) to[bend left=24] (b) (c) to[bend right=12] (b);
		\draw (b) to[bend left=12] (a) (b) to[bend left=24] (a) (b) to[bend right=12] (a);
		\draw (c) to[bend left=12] (f) (c) to[bend left=24] (f) (c) to[bend right=12] (f);
		\draw (d) to[bend left=12] (g) (d) to[bend right=24] (g) (d) to[bend right=12] (g);
		\draw (e) to[bend left=12] (g) (e) to[bend left=24] (g) (e) to[bend right=12] (g);
		\draw (f) to[bend left=12] (g) (f) to[bend left=24] (g) (f) to[bend right=12] (g);
		\draw (a) circle(1); \draw (b) circle(1); \draw (c) circle(1);
		\draw (d) circle(1); \draw (e) circle(1); \draw (f) circle(1);
		\draw (g) circle(1);
		\tikzstyle{every node}=[draw, shape=circle, minimum size=1pt, inner sep=0.75pt, fill=red]
		\node at (-1.01,0.18) {}; \node at (-1.01,-0.18) {}; \node at (-1.06,-0.36) {};
		\node at (1.01,0.18) {}; \node at (1.01,-0.18) {}; \node at (1.06,-0.36) {};
		\node at (-1.18,0.58) {}; \node at (-1.4,0.81) {}; \node at (-1.55,0.9) {};
		\node at (1.18,0.58) {}; \node at (1.4,0.81) {}; \node at (1.55,0.9) {};

		\node at (-0.17,2.99) {}; \node at (0.17,2.99) {}; \node at (-0.34,2.96) {};
		\node at (-0.57,1.19) {}; \node at (-0.8,1.41) {}; \node at (-0.89,1.55) {};
		\node at (0.57,1.19) {}; \node at (0.8,1.41) {}; \node at (0.89,1.55) {};
		\node at (-0.17,5.99) {}; \node at (0.17,5.99) {}; \node at (-0.34,5.96) {};
		\node at (-0.17,4.01) {}; \node at (0.17,4.01) {}; \node at (-0.34,4.04) {};

		\node at (-0.17,7.01) {}; \node at (0.17,7.01) {}; \node at (-0.34,7.04) {};
		\node at (-0.57,8.81) {}; \node at (-0.8,8.59) {}; \node at (-0.89,8.45) {};
		\node at (0.57,8.81) {}; \node at (0.8,8.59) {}; \node at (0.89,8.45) {};
		\node at (1.57,9.1) {}; \node at (1.42,9.19) {}; \node at (1.19,9.41) {};
		\node at (-1.57,9.1) {}; \node at (-1.42,9.19) {}; \node at (-1.19,9.41) {};
	\end{tikzpicture}
	\quad\raise20ex\hbox{as:}\quad
	\begin{tikzpicture}[scale=0.6]
		\draw[draw=none,fill=white!70!yellow] (0,5) ellipse (1.1 and 3.5);
		\tikzstyle{every node}=[draw, shape=circle, minimum size=2pt,inner sep=1.4pt, fill=black]
		\tikzstyle{every path}=[color=black]
		\node at (-2,0) (a) {}; \node at (2,0) (b) {}; \node at (0,2) (c) {};
		\node[fill=white!80!black] at (-2,10) (d) {}; \node at (2,10) (e) {};
		\node[fill=white!80!black] at (0,5) (f) {};
		\node at (0,8) (g) {};
		\tikzstyle{every path}=[color=blue,thick]
		\draw (a) -- (b) -- (c) -- (a);
		\draw[dotted] (d) -- (g) -- (f) -- (c);
		\draw (g) -- (e);
		\tikzstyle{every path}=[color=green!50!black,dotted]
		\draw (d) -- (a) to[bend left=18] (f) (f) to[bend left=18] (b) (b) -- (e);
		\draw (c) to[bend left=55] (g);
		\draw (c) to[bend right=30] (f);
		\draw[dotted] (c) -- (f);
		\tikzstyle{every path}=[color=red!70!black]
		\draw (a) to[bend left=12] (c) (a) to[bend left=24] (c) (a) to[bend right=12] (c);
		\draw (c) to[bend left=12] (b) (c) to[bend left=24] (b) (c) to[bend right=12] (b);
		\draw (b) to[bend left=12] (a) (b) to[bend left=24] (a) (b) to[bend right=12] (a);
		\draw (c) to[bend left=12] (f) (c) to[bend left=24] (f) (c) to[bend right=12] (f);
		\draw[dotted] (d) to[bend left=12] (g) (d) to[bend right=24] (g) (d) to[bend right=12] (g);
		\draw (e) to[bend left=12] (g) (e) to[bend left=24] (g) (e) to[bend right=12] (g);
		\draw (f) to[bend left=12] (g) (f) to[bend left=24] (g) (f) to[bend right=12] (g);
		\draw (a) circle(1); \draw (b) circle(1); \draw (c) circle(1);
		\draw[dotted] (d) circle(1); \draw (e) circle(1); \draw (f) circle(1);
		\draw (g) circle(1);
		\tikzstyle{every node}=[draw, shape=circle, minimum size=1pt, inner sep=0.75pt, fill=red]
		\node at (-1.01,0.18) {}; \node at (-1.01,-0.18) {}; \node at (-1.06,-0.36) {};
		\node at (1.01,0.18) {}; \node at (1.01,-0.18) {}; \node at (1.06,-0.36) {};
		\node at (-1.18,0.58) {}; \node at (-1.4,0.81) {}; \node at (-1.55,0.9) {};
		\node at (1.18,0.58) {}; \node at (1.4,0.81) {}; \node at (1.55,0.9) {};

		\node at (-0.17,2.99) {}; \node at (0.17,2.99) {}; \node at (-0.34,2.96) {};
		\node at (-0.57,1.19) {}; \node at (-0.8,1.41) {}; \node at (-0.89,1.55) {};
		\node at (0.57,1.19) {}; \node at (0.8,1.41) {}; \node at (0.89,1.55) {};
		\node at (-0.17,5.99) {}; \node at (0.17,5.99) {}; \node at (-0.34,5.96) {};
		\node at (-0.17,4.01) {}; \node at (0.17,4.01) {}; \node at (-0.34,4.04) {};

		\node at (-0.17,7.01) {}; \node at (0.17,7.01) {}; \node at (-0.34,7.04) {};
		\node at (0.57,8.81) {}; \node at (0.8,8.59) {}; \node at (0.89,8.45) {};
		\node at (1.57,9.1) {}; \node at (1.42,9.19) {}; \node at (1.19,9.41) {};
		\tikzstyle{every node}=[draw,dotted, shape=circle, minimum size=1pt, inner sep=0.75pt, fill=white!80!red]
		\node at (-0.57,8.81) {}; \node at (-0.8,8.59) {}; \node at (-0.89,8.45) {};
		\node at (-1.57,9.1) {}; \node at (-1.42,9.19) {}; \node at (-1.19,9.41) {};
	\end{tikzpicture}
	\end{center}
		\caption{An illustration of the framing topological minor relation; the framing $\bar G_1$ (of the $5$-vertex \ppdg graph $(G_1,\ca H_1)$) on the left is embedded in the framing $\bar G_2$ (of the $7$-vertex graph $(G_2,\ca H_2)$) in the middle,
		and this embedding is emphasised as a topological minor in the picture on the right.
		Notice that the framing triplet in $\bar G_1$ highlighted in the left picture with yellow background is mapped (as three internally disjoint red paths) into a union of two framing triplets plus the intermediate framing cycle in $\bar G_2$, as highlighted with yellow background in the picture on the right.
		}
		\label{fig:framingminor}
	\end{figure}

	However, to state the desired characterisation we still need to technically generalise Definition~\ref{def:framing} to an {\em extended framing} of a \ppdg graph $(G,\ca H)$ which, informally, allows us to use possible additional connector vertices and arbitrary connector edges between the components of $\ca H$.
	\ifshort See the preprint paper for all details.\fi
	\iflong See the Appendix for all details.\fi
	\begin{restatable}{lemma}{encodeintopol}\label{lem:restrtopol}\apxmark
		Let $\ca P_1=(G_1,\ca{H}_1)$ and $\ca P_2=(G_2,\ca{H}_2)$ be \ppdg graphs where $\ca{H}_1$ and $\ca{H}_2$ are plane.
		Let $\bar G_2$ be a framing of $\ca P_2$.
		Then some subdivision of $\ca P_1$ is a \ppdg subgraph of $\ca P_2$, if and only if there exists an extended framing $\bar G_1$ of $\ca P_1$ such that $\bar G_1$ is a restricted topological minor of~$\bar G_2$.
	\end{restatable}

	We now finish a proof sketch of Lemma~\ref{lem:MSOforobstruction} easily. 
	Let $\cf F$ be the finite set of all distinct extended framings of $\ca P_1$.
	Using Lemma~\ref{lem:restrtopol}, we may write the formula $\sigma\equiv\bigvee_{\bar G_1\in\cf F}\sigma[\bar G_1]$ where $\bar G_2\models\sigma[\bar G_1]$ routinely expresses that $\bar G_1$ is a framing topological minor of~$\bar G_2$
	(this description uses auxiliary precomputed labels distinguishing the types of edges in~$\bar G_2$).

	We also need to address the other kind of obstruction in Theorem~\ref{thm:PEG-relaxed-obstructions} with the following:
	\begin{restatable}{lemma}{encodealternating}\label{lem:MSOalternating}\apxmark
		There exists an MSO$_2$-formula~$\tau$ such that the following is true:
		\begin{itemize}\item
		For any \ppdg graph~$\ca P_2=(G_2,\ca{H}_2)$ and any framing $\bar G_2$ of $\ca P_2$ we have that $\bar G_2\models\tau$,
		if and only if there exists an alternating chain in $\ca P_2$.
		\end{itemize}
	\end{restatable}\medskip

	Now we can sketch a proof of the key Lemma~\ref{lem:phaseIIi} which we reformulate slightly for clarity:%
	\begin{restatable}[Lemma~\ref{lem:phaseIIi}]{lemma}{MSOphaseII} \label{lem:phaseII}
		For every $k\geq0$ there is an MSO\(_2\)-formula $\psi_k$ such that the following holds.
		Given a \ppdg graph $\ca P$, with some edges of $\ca P$ marked as `uncrossable', one can in polynomial time construct a graph $G'$ such that $G'\models\psi_k$ if and only if there exists a \crconf{$k$} drawing of~$\ca P$.
	\end{restatable}

	\begin{proof}[Sketch of proof]
		Recall that we may assume $\ca H$ to be a plane graph.
		We first give a rough outline of what we want to achieve and then sketch the core steps of the proof.

		The graph $G'$ will be based on a framing (as used above).
		Imagine a conforming drawing $\ca G$ of $G$ (extending $\ca H$) with $\crd(\ca G)=k$ and its planarisation $\ca G^\times$.
		If we were able to ``guess'', within the formula $\psi_k$, the additional $k$ vertices (those of $\ca G^\times$) making the crossings, then we would finish by checking \ppd planarity of the result (i.e., of the guessed~$\ca G^\times$).
		Using Theorem~\ref{thm:PEG-relaxed-obstructions}, the latter would follow by an application of Lemmas~\ref{lem:MSOforobstruction} and~\ref{lem:MSOalternating}.

		Specifically, for the task of ``guessing the crossings'', we subdivide each edge of $\ca P$ which is not marked as `uncrossable' by $k$ new vertices, called {\em auxiliary vertices} of this \ppdg subdivision $\ca P_0=(G_0,\ca{H}_0)$ of~$\ca P$.
		A subdivision clearly does not change the crossing number; $\crg(\ca P)=\crg(\ca P_0)$.
		Then we interpret ``guessing a crossing'' in $\ca P_0$ as picking (with existential quantifiers in~$\psi_k$) a pair $r_1',r_1''\in V(G_0)\setminus V(G)$ of auxiliary vertices such that not both $r_1'$ and $r_1''$ are from edges of~$H$, and identifying~$r_1'=r_1''$.
		Let $\ca P_0[r_1'=r_1'']$ denote the graph after such an identification.
		Note that since we do not identify auxiliary pairs from two edges of $H$, the following holds---if $\bar G_0$ is a framing of $\ca P_0$, then $\bar G_0[r_1'=r_1'']$ is a graph isomorphic to the corresponding framing of $\ca P_0[r_1'=r_1'']$.
	
		We let $G'=\bar G_0$ be a framing of $\ca P_0=(G_0,\ca{H}_0)$.
		Let $\vect r'=(r_i':i\in[k])$ and $\vect r''=(r_i'':i\in[k])$ be two $k$-tuples of vertex variables (which are used to specify the $k$ identifications of vertex pairs in~$\ca P_0[\vect r'=\vect r'']$).
		We write the desired formula as
		\begin{align*} \psi_k\>\equiv\> \exists\,\vect r',\vect r''\Big( \bigwedge\nolimits_{r,s\in\vect r'\cup\vect r''}r\not=s \wedge \bigwedge\nolimits_{i\in[k]}\chi(r_i',r_i'') \wedge\> \psi'_k[\vect r',\vect r''] \Big), \end{align*}
		where $\chi(r_i',r_i'')$ checks that $r_i',r_i''$ are auxiliary vertices and not both coming from edges of~$H$ (using precomputed labels of the auxiliary vertices).
		The formula $\psi'_k[\vect r',\vect r'']$ then tests whether the \ppdg graph $\ca P_0[\vect r'=\vect r'']$ admits a planar drawing extending~$\ca H_0$.
		This is a technical task based on Lemmas~\ref{lem:MSOforobstruction} and~\ref{lem:MSOalternating}, and we leave full details for the
		\ifshort preprint paper.\fi
		\iflong Appendix.\fi
	\end{proof}

	Finally, we summarise how {\bf Theorem~\ref{thm:main}} follows from the previous claims.
	Given a \ppdg graph $(G,\ca H)$ and an integer $k>0$, we first make $\ca H$ planarised.
	Then, using Theorem~\ref{thm:reducegrah}, we either conclude that $\crgpd(G,\ca H)>k$, or we iteratively reduce the input to an equivalent instance $(G',\ca H')$ with the same solution value~$k$.
	Moreover, using also Lemma~\ref{lem:tw}, we have that the tree-width of any framing $\bar G'$ of $(G',\ca H')$ is bounded in terms of~$k$.
	We can hence efficiently decide whether $\crgpd(G',\ca H')\leq k$ using Courcelle's theorem applied with the formula $\psi_k$ from Lemma~\ref{lem:phaseII} to a framing $\bar G'$ of $(G',\ca H')$.
	
	{\bf(*)} {We can also observe that the \FPT runtime of this procedure is \(\mathcal{O}(f(k)\cdot|V(G)|^3)\).}

	\section{Restricting crossings per edge}
	\label{sec:per-edge}
	Next we outline some nice consequences of our techniques for previously considered drawing extension settings.
	Firstly, we are able to trivially modify our \FPT-algorithm for \textsc{\PPDCR} by additionally encoding the fact that in a solution every edge in \(E(G) \setminus E(H)\) has at most \(c\) crossings by introducing \(k\) auxiliary vertices for each edge in \(E(H)\), but only \(\min\{c,k\}\) auxiliary vertices for each edge in \(E(G) \setminus E(H)\) in the proof of \Cref{lem:phaseII}.
	This immediately gives us {Theorem~\ref{thm:cplanar-improve}} restated from above.
 	\cPlanar*

\smallskip	
	Another closely related problem that has been considered in literature asks for the smallest number of non-predrawn crossings in a \emph{simple} drawing that coincides with the given \ppdg\ graph, in which each edge in \(E(G) \setminus E(H)\) has at most \(c\) crossings.
	I.e., compared to \textsc{\PPDCRc{\(c\)}} we only allow drawings in which no pair of edges crosses more than once (crossings between adjacent edges can always be avoided).
	The difficulty for our approach here is that we need to record the information of which edges in \(\mathcal{H}^\times\) correspond to the same edge in the non-planarised \pdsg \(H\) (this part can be handled by an MSO$_2$-formula with help of special edge labels, cf.~\cite{GanianHKPV21}),
	and more importantly to keep this information, even during our iterative reduction of \(G\) and \(\mathcal{H}^\times\) described in \Cref{sec:phase1}.
	The latter seems to be a deep problem, not easy to overcome and a good direction for continuing research.

	Nevertheless, using the more restrictive parameterisation by \(|E(G) \setminus E(H)| + c\) (which also naturally bounds the crossing number), we are able to give an improvement on the best known result in~\cite{GanianHKPV21}:
	finding the least number of crossings in a simple drawing which coincides with the given partial drawing and in which each edge outside of the \pdsg has at most \(c\) crossings in \FPT-time.
	The known result assumes that the planarised \pdsg is connected, an assumption that we can easily drop using our MSO\(_2\)-encoding in combination with a crucial structural lemma which we adapt from \cite{GanianHKPV21} to `stitch' together relevant edges in \(\mathcal{H}^\times\) that correspond to the same edge in \(H\).
	This improvement over~\cite{GanianHKPV21} results in {\bf Theorem~\ref{thm:Simple}} stated in the Introduction.

	\section{Conclusion}\label{sec:conclu}

	To summarise, we have shown that some algorithmic results for the classical crossing-number can be extended to the \ppd setting, similarly to the respective planarity question~\cite{AngeliniDFJVKPR15}.
	However, what can we say about structural properties of the \ppdcr?

	For instance, what can we say about the minimal graphs of a certain crossing-number value?
	We call a \ppdg graph $\ca P=(G,\ca H)$ {\em $k$-crossing-critical} if the \ppdcr of $\ca P$ is at least~$k$, but this crossing number drops down below $k$ after deleting any edge, predrawn or not, from~$\ca P$
	(alternatively, one may also include removing any edge from $H$ while keeping it in~$G$ to the definition).
	We have recently gotten a complete rough asymptotical characterisation of classical $k$-crossing-critical graphs~\cite{DBLP:conf/compgeom/DvorakHM18}, but here we see an important difference in behaviour.
	For classical $k$-crossing-critical graphs, optimal drawings (i.e.\ those achieving the minimum number of crossings) can never contain a collection of edge-disjoint cycles drawn nested in each other and of size arbitrarily large compared to~$k$
	(this is implicit in \cite{DBLP:journals/jct/Hernandez-VelezST12} or \cite{DBLP:conf/compgeom/DvorakHM18}).
	In contrast to that, we provide:
	\begin{restatable}{proposition}{baddualdiam}\apxmark \label{pro:baddualdiam}
		For each $k\geq8$ and $m>0$, there exists a \ppdg graph $\ca P=(G,\ca{H})$ such that $\ca P$ is $k$-crossing-critical and that an optimal (with minimum crossings) drawing of $\ca P$ extending $\ca H$ contains at least $m$ vertex-disjoint nested cycles from~$G-E(H)$.
	\end{restatable}

	Consequently, even a rough characterisation of \ppdg $k$-crossing-critical graphs is a widely open question worth further investigation.
	Unfortunately, already at the starting point of this track we lack a good analogue of the result \cite{DBLP:journals/jct/RichterT93}, saying that a $k$-crossing-critical graph has its crossing number bounded in terms of~$k$, whose proof simply breaks down in the \ppd setting.
	Having a result like \cite{DBLP:journals/jct/RichterT93} in the predrawn setting we could, as a first step, adapt the arguments from Section~\ref{sec:PPDCR} to prove that \ppdg $k$-crossing-critical graphs have treewidth bounded in terms of~$k$.

	\bibliography{references}

\begin{thebibliography}{10}

\bibitem{AngeliniDFJVKPR15}
Patrizio Angelini, Giuseppe Di~Battista, Fabrizio Frati, V\'{\i}t Jel\'{\i}nek,
  Jan Kratochv\'{\i}l, Maurizio Patrignani, and Ignaz Rutter.
\newblock Testing planarity of partially embedded graphs.
\newblock {\em ACM Trans. Algorithms}, 11(4), April 2015.
\newblock \href {https://doi.org/10.1145/2629341} {\path{doi:10.1145/2629341}}.

\bibitem{ArnborgLS91}
Stefan Arnborg, Jens Lagergren, and Detlef Seese.
\newblock Easy problems for tree-decomposable graphs.
\newblock {\em J. Algorithms}, 12(2):308--340, 1991.
\newblock \href {https://doi.org/10.1016/0196-6774(91)90006-K}
  {\path{doi:10.1016/0196-6774(91)90006-K}}.

\bibitem{Cabello13}
Sergio Cabello.
\newblock Hardness of approximation for crossing number.
\newblock {\em Discrete Comput. Geom.}, 49(2):348--358, March 2013.

\bibitem{CabelloM13}
Sergio Cabello and Bojan Mohar.
\newblock Adding one edge to planar graphs makes crossing number and
  1-planarity hard.
\newblock {\em SIAM J. Comput.}, 42(5):1803--1829, January 2013.

\bibitem{CaselFMMS21}
Katrin Casel, Henning Fernau, Mehdi {Khosravian Ghadikolaei}, Jérôme Monnot,
  and Florian Sikora.
\newblock On the complexity of solution extension of optimization problems.
\newblock {\em Theoretical Computer Science}, 2021.
\newblock \href {https://doi.org/https://doi.org/10.1016/j.tcs.2021.10.017}
  {\path{doi:https://doi.org/10.1016/j.tcs.2021.10.017}}.

\bibitem{DBLP:journals/jacm/ChekuriC16}
Chandra Chekuri and Julia Chuzhoy.
\newblock Polynomial bounds for the grid-minor theorem.
\newblock {\em J. {ACM}}, 63(5):40:1--40:65, 2016.

\bibitem{DBLP:conf/compgeom/ChimaniH16}
Markus Chimani and Petr Hlin\v{e}n{\'{y}}.
\newblock Inserting multiple edges into a planar graph.
\newblock In {\em SoCG}, volume~51 of {\em LIPIcs}, pages 30:1--30:15. Schloss
  Dagstuhl - Leibniz-Zentrum f{\"{u}}r Informatik, 2016.

\bibitem{Courcelle90}
Bruno Courcelle.
\newblock The monadic second-order logic of graphs. {I}. recognizable sets of
  finite graphs.
\newblock {\em Inf. Comput.}, 85(1):12--75, 1990.
\newblock \href {https://doi.org/10.1016/0890-5401(90)90043-H}
  {\path{doi:10.1016/0890-5401(90)90043-H}}.

\bibitem{CyganFKLMPPS15}
Marek Cygan, Fedor~V. Fomin, \L{}ukasz Kowalik, Daniel Lokshtanov, D{\'{a}}niel
  Marx, Marcin Pilipczuk, Micha\l{} Pilipczuk, and Saket Saurabh.
\newblock {\em Parameterized Algorithms}.
\newblock Springer, 2015.
\newblock \href {https://doi.org/10.1007/978-3-319-21275-3}
  {\path{doi:10.1007/978-3-319-21275-3}}.

\bibitem{Diestel12}
Reinhard Diestel.
\newblock {\em Graph Theory, 4th Edition}, volume 173 of {\em Graduate texts in
  mathematics}.
\newblock Springer, 2012.

\bibitem{DowneyF13}
Rodney~G. Downey and Michael~R. Fellows.
\newblock {\em Fundamentals of Parameterized Complexity}.
\newblock Texts in Computer Science. Springer, 2013.
\newblock \href {https://doi.org/10.1007/978-1-4471-5559-1}
  {\path{doi:10.1007/978-1-4471-5559-1}}.

\bibitem{DBLP:conf/compgeom/DvorakHM18}
Zdenek Dvo\v{r}{\'{a}}k, Petr Hlin\v{e}n{\'{y}}, and Bojan Mohar.
\newblock Structure and generation of crossing-critical graphs.
\newblock In {\em SoCG}, volume~99 of {\em LIPIcs}, pages 33:1--33:14. Schloss
  Dagstuhl - Leibniz-Zentrum f{\"{u}}r Informatik, 2018.

\bibitem{EibenGHK21}
Eduard Eiben, Robert Ganian, Thekla Hamm, and O~joung Kwon.
\newblock Measuring what matters: A hybrid approach to dynamic programming with
  treewidth.
\newblock {\em Journal of Computer and System Sciences}, 121:57--75, 2021.
\newblock \href {https://doi.org/https://doi.org/10.1016/j.jcss.2021.04.005}
  {\path{doi:https://doi.org/10.1016/j.jcss.2021.04.005}}.

\bibitem{EibenGHKN20b}
Eduard Eiben, Robert Ganian, Thekla Hamm, Fabian Klute, and Martin
  N{\"{o}}llenburg.
\newblock Extending nearly complete 1-planar drawings in polynomial time.
\newblock In Javier Esparza and Daniel Kr{\'{a}}l', editors, {\em {MFCS} 2020},
  volume 170 of {\em LIPIcs}, pages 31:1--31:16. Schloss Dagstuhl -
  Leibniz-Zentrum f{\"{u}}r Informatik, 2020.
\newblock \href {https://doi.org/10.4230/LIPIcs.MFCS.2020.31}
  {\path{doi:10.4230/LIPIcs.MFCS.2020.31}}.

\bibitem{EibenGHKN20a}
Eduard Eiben, Robert Ganian, Thekla Hamm, Fabian Klute, and Martin
  N{\"{o}}llenburg.
\newblock Extending partial 1-planar drawings.
\newblock In Artur Czumaj, Anuj Dawar, and Emanuela Merelli, editors, {\em
  {ICALP} 2020}, volume 168 of {\em LIPIcs}, pages 43:1--43:19. Schloss
  Dagstuhl - Leibniz-Zentrum f{\"{u}}r Informatik, 2020.
\newblock \href {https://doi.org/10.4230/LIPIcs.ICALP.2020.43}
  {\path{doi:10.4230/LIPIcs.ICALP.2020.43}}.

\bibitem{FlumGrohe06}
J\"{o}rg Flum and Martin Grohe.
\newblock {\em Parameterized Complexity Theory}, volume XIV of {\em Texts in
  Theoretical Computer Science. An EATCS Series}.
\newblock Springer, 2006.
\newblock \href {https://doi.org/10.1007/3-540-29953-X}
  {\path{doi:10.1007/3-540-29953-X}}.

\bibitem{GanianHKPV21}
Robert Ganian, Thekla Hamm, Fabian Klute, Irene Parada, and Birgit Vogtenhuber.
\newblock Crossing-optimal extension of simple drawings.
\newblock In Nikhil Bansal, Emanuela Merelli, and James Worrell, editors, {\em
  {ICALP} 2021}, volume 198 of {\em LIPIcs}, pages 72:1--72:17. Schloss
  Dagstuhl - Leibniz-Zentrum f{\"{u}}r Informatik, 2021.
\newblock \href {https://doi.org/10.4230/LIPIcs.ICALP.2021.72}
  {\path{doi:10.4230/LIPIcs.ICALP.2021.72}}.

\bibitem{GareyJ83}
Michael~R. Garey and David~S. Johnson.
\newblock Crossing number is {NP-complete}.
\newblock {\em SIAM J. Algebr. Discrete Methods}, 4(3):312--316, September
  1983.

\bibitem{Grohe04}
Martin Grohe.
\newblock Computing crossing numbers in quadratic time.
\newblock {\em J. Comput. Syst. Sci.}, 68(2):285--302, 2004.
\newblock \href {https://doi.org/10.1016/j.jcss.2003.07.008}
  {\path{doi:10.1016/j.jcss.2003.07.008}}.

\bibitem{DBLP:journals/jct/Hernandez-VelezST12}
C{\'{e}}sar Hern{\'{a}}ndez{-}V{\'{e}}lez, Gelasio Salazar, and Robin Thomas.
\newblock Nested cycles in large triangulations and crossing-critical graphs.
\newblock {\em J. Comb. Theory, Ser. {B}}, 102(1):86--92, 2012.

\bibitem{Hlineny06}
Petr Hlin\v{e}n{\'{y}}.
\newblock Crossing number is hard for cubic graphs.
\newblock {\em Journal of Comb. Theory, Ser. B}, 96(4):455--471, 2006.
\newblock \href {https://doi.org/https://doi.org/10.1016/j.jctb.2005.09.009}
  {\path{doi:https://doi.org/10.1016/j.jctb.2005.09.009}}.

\bibitem{HopcroftT74}
John Hopcroft and Robert Tarjan.
\newblock Efficient planarity testing.
\newblock {\em J. ACM}, 21(4):549–568, oct 1974.
\newblock \href {https://doi.org/10.1145/321850.321852}
  {\path{doi:10.1145/321850.321852}}.

\bibitem{JelinekKR13}
Vít Jelínek, Jan Kratochvíl, and Ignaz Rutter.
\newblock A {K}uratowski-type theorem for planarity of partially embedded
  graphs.
\newblock {\em Computational Geometry}, 46(4):466--492, 2013.
\newblock SoCG 2011.
\newblock \href {https://doi.org/https://doi.org/10.1016/j.comgeo.2012.07.005}
  {\path{doi:https://doi.org/10.1016/j.comgeo.2012.07.005}}.

\bibitem{KawarabayashiR07}
Ken-ichi Kawarabayashi and Bruce Reed.
\newblock Computing crossing number in linear time.
\newblock In {\em Proceedings of the Thirty-Ninth Annual ACM Symposium on
  Theory of Computing}, STOC '07, page 382–390. Association for Computing
  Machinery, 2007.
\newblock \href {https://doi.org/10.1145/1250790.1250848}
  {\path{doi:10.1145/1250790.1250848}}.

\bibitem{Kloks94}
Ton Kloks.
\newblock {\em Treewidth: Computations and Approximations}.
\newblock 1994.

\bibitem{Korhonen21}
Tuukka Korhonen.
\newblock Single-exponential time 2-approximation algorithm for treewidth.
\newblock {\em CoRR}, abs/2104.07463, 2021.
\newblock \href {http://arxiv.org/abs/2104.07463} {\path{arXiv:2104.07463}}.

\bibitem{MisueELS95}
Kazuo Misue, Peter Eades, Wei Lai, and Kozo Sugiyama.
\newblock Layout adjustment and the mental map.
\newblock {\em Journal of Visual Languages and Computing}, 6(2):183--210, 1995.
\newblock \href {https://doi.org/https://doi.org/10.1006/jvlc.1995.1010}
  {\path{doi:https://doi.org/10.1006/jvlc.1995.1010}}.

\bibitem{mtbook}
Bojan Mohar and Carsten Thomassen.
\newblock {\em Graphs on Surfaces}.
\newblock Johns Hopkins series in the mathematical sciences. Johns Hopkins
  University Press, 2001.

\bibitem{DBLP:journals/algorithmica/PelsmajerSS11}
Michael~J. Pelsmajer, Marcus Schaefer, and Daniel Stefankovic.
\newblock Crossing numbers of graphs with rotation systems.
\newblock {\em Algorithmica}, 60(3):679--702, 2011.

\bibitem{DBLP:journals/jct/RichterT93}
Robert~B. Richter and Carsten Thomassen.
\newblock Minimal graphs with crossing number at least \emph{k}.
\newblock {\em J. Comb. Theory, Ser. {B}}, 58(2):217--224, 1993.

\bibitem{RobertsonST94}
Neil Robertson, Paul Seymour, and Robin Thomas.
\newblock Quickly excluding a planar graph.
\newblock {\em Journal of Comb. Theory, Ser. B}, 62(2):323--348, 1994.
\newblock \href {https://doi.org/https://doi.org/10.1006/jctb.1994.1073}
  {\path{doi:https://doi.org/10.1006/jctb.1994.1073}}.

\bibitem{Veblen05}
Oswald Veblen.
\newblock Theory on plane curves in non-metrical analysis situs.
\newblock {\em Trans. Am. Math. Soc.}, 6(1):83--98, 1905.

\bibitem{Wagner37}
Klaus Wagner.
\newblock {\"U}ber eine {E}igenschaft der ebenen {K}omplexe.
\newblock {\em Mathematische Annalen}, 114:570--590, 1937.

\bibitem{WeiKuanW99}
Shih Wei-Kuan and Hsu Wen-Lian.
\newblock A new planarity test.
\newblock {\em Theoretical Computer Science}, 223(1):179--191, 1999.
\newblock \href {https://doi.org/https://doi.org/10.1016/S0304-3975(98)00120-0}
  {\path{doi:https://doi.org/10.1016/S0304-3975(98)00120-0}}.

\end{thebibliography}
	
\iflong
	\newpage
	\appendix

	\section{Additions to Section~\ref{sec:prelims}}
	\subsection*{Parameterised complexity, treewidth and grids}
	In parameterised complexity~\cite{CyganFKLMPPS15,DowneyF13,FlumGrohe06},
	the complexity of a problem is studied not only with respect to the input size, but also with respect to some problem parameter(s). The core idea behind parameterised complexity is that the combinatorial explosion resulting from the \NP-hardness of a problem can sometimes be confined to certain structural parameters that are small in practical settings. We now proceed to the formal definitions.
	
	A {\it parameterised problem} $Q$ is a subset of $\Omega^* \times
	\mathbb{N}$, where $\Omega$ is a fixed alphabet. Each instance of $Q$ is a pair $(I, \kappa)$, where $\kappa \in \mathbb{N}$ is called the {\it
		parameter}. A parameterised problem $Q$ is
	{\it fixed-parameter tractable} (\FPT)~\cite{FlumGrohe06,DowneyF13,CyganFKLMPPS15}, if there is an
	algorithm, called an {\em \FPT-algorithm},  that decides whether an input $(I, \kappa)$
	is a member of $Q$ in time $f(\kappa) \cdot |I|^{\mathcal{O}(1)}$, where $f$ is a computable function and $|I|$ is the input instance size.  The class \FPT{} denotes the class of all fixed-parameter tractable parameterised problems.
	A parameterised problem $Q$
	is {\it \FPT-reducible} to a parameterised problem $Q'$ if there is
	an algorithm, called an \emph{\FPT-reduction}, that transforms each instance $(I, \kappa)$ of $Q$
	into an instance $(I', \kappa')$ of
	$Q'$ in time $f(\kappa)\cdot |I|^{\mathcal{O}(1)}$, such that $\kappa' \leq g(\kappa)$ and $(I, \kappa) \in Q$ if and
	only if $(I', \kappa') \in Q'$, where $f$ and $g$ are computable
	functions. 
	
	An extremely popular parameter for graph problems is \emph{treewidth}.
	A \emph{tree-decomposition}~$\mathcal{T}$ of a graph $G=(V,E)$ is a pair 
	$(T,\chi)$, where $T$ is a tree (whose vertices we call \emph{nodes}) rooted at a node $r$ and $\chi$ is a function that assigns each node $t$ a set $\chi(t) \subseteq V$ such that the following holds: 
	\begin{itemize}
		\item For every $uv \in E$ there is a node
		$t$ such that $u,v\in \chi(t)$.
		\item For every vertex $v \in V$,
		the set of nodes $t$ satisfying $v\in \chi(t)$ forms a subtree of~$T$.
	\end{itemize}
	
	The \emph{width} of a tree-decomposition $(T,\chi)$ is the size of a largest set $\chi(t)$ minus~$1$, and the \emph{treewidth} of the graph $G$,
	denoted $\operatorname{tw}(G)$, is the minimum width of a tree-decomposition of~$G$.
	Efficient fixed-parameter algorithms are known for computing a tree-decomposition of near-optimal width~\cite{Korhonen21,Kloks94}.

	A {\em $k\times k$ square grid} is the planar graph $D$ on the vertex set $V(D)=\{(i,j):i,j\in[k]\}$ and the edge set $E(D)=\{\{(i,j),(i',j')\}:|i-i'|+|j-j'|=1\}$.
	It is well known that if a graph $G$ contains a $k\times k$ square grid minor, then $\operatorname{tw}(G)\geq k$ and, conversely, there is a polynomial function $f(n)$ such that if $k$ is the largest integer for which $G$ contains a $k\times k$ square grid minor, then $\operatorname{tw}(G)\leq f(k)$ \cite{DBLP:journals/jacm/ChekuriC16}.

	\smallskip		
	To keep our exposition self-contained, we also include a brief description of the expressive power of the {\em MSO$_2$ logic of graphs}.
	This logic is defined over graphs $G$ with the vertex set $V(G)$ and the edge set $E(G)$.\footnote{This is different from related MSO$_1$ logic in which the edges of $G$ are expressed only as a binary predicate and, consequently, MSO$_1$ has weaker expressive power towards sets of edges.}
	An MSO$_2$-formula can use (and quantify) variables for vertices $w\in V(G)$ and edges $f\in E(G)$, and for their sets $W\subseteq V(G)$ and $F\subseteq E(G)$.
	Then there is the standard equality predicate $=$ and the incidence predicate $\prebox{inc}(w,f)$ expressing that $w$ is an end of an edge~$f$.
	Additionally, one can assign arbitrary unary predicates as labels (or colours) to the vertices and edges of $G$ and access the labels within the formula.
	The famous theorem of Courcelle~\cite{Courcelle90} states that any decision property expressible in MSO$_2$ logic can be decided by an \FPT-algorithm on graphs of bounded treewidth, where the parameter is the sum of a formula length and the value of treewidth.

        \subsection*{Grohe's result for the classical crossing number}

	\begin{definition}[Contraction of a subdrawing]
		Given a drawing \(\mathcal{G}\) of a graph \(G\), and a subgraph \(H \subseteq G\) none of whose edges are crossed in \(\mathcal{G}\), \emph{contracting} \(H\) in \(\mathcal{G}\) defines a drawing of the graph which arises from \(G\) by contracting \(H\) in the following way.
		\begin{itemize}
			\item Every vertex in \(V(G) \setminus V(H)\) and every edge in \(E(G) \setminus E(H)\) is mapped to the same point or curve as \(\mathcal{G}\) does.
			\item The vertex \(v_H\) to which \(H\) is contracted is mapped to an arbitrary but fixed point on the convex hull of the drawing (according to \(\mathcal{G}\)) of \(H\) which is not a crossing in \(\mathcal{G}\).
			\item It remains to define the drawing of the edges \(e\) incident to \(v_H\), each of which that corresponds to an edge \(uv \in E(G)\) where \(u \in V(G) \setminus V(H)\) and \(v \in V(H)\).
			Consider the simple curve \(\gamma\) arising from the drawing (according to \(\mathcal{G}\)) of \(uv\) capped at its first intersection with the boundary of the convex hull of the drawing of \(H\) in \(\mathcal{G}\) -- in case this intersection does not exist, we simply take the whole curve.
			Now we map \(e\) to the concatenation of \(\gamma\) with the straight line between the endpoint of \(\gamma\) and the drawing of \(v_H\).
		\end{itemize}
	\end{definition}

	The following is an abstracted and condensed formulation of more technical lemmas and proofs in the original paper~\cite{Grohe04}.
	Recall that a drawing is called {\em conforming} if no edge marked `uncrossable' is crossed in it.
	\begin{theorem}[adapted from {\cite[Proofs of Lemmas~6 and 7]{Grohe04}}]\label{thm:grohegrid}
		For all \(k \in \mathbb{N}\) there is some \(w \in \mathbb{N}\) such that, given a graph \(G\) in which some edges are marked `uncrossable', in \FPT-time parameterised by \(k\) we can
		\begin{enumerate}
			\item decide that the number of crossings in any conforming drawing of \(G\) is more than \(k\); or
			\item find a tree decomposition of \(G\) of width at most \(w\); or
			\item find a connected subgraph \(I \subseteq G\) with \(|V(I)| \geq 6\) and a cycle \(C\) in \(G - V(I)\) that contains the neighbourhood of \(I\) such that; if in the graph \(G'\) that arises from \(G\) by contracting \(I\) to the vertex \(v_I\) we mark as `uncrossable' all edges that are in \(C\), or incident to \(v_I\), or already in \(E(G)\) and marked as `uncrossable', then the following holds:
			\begin{enumerate}
				\item In any conforming drawing of \(G\) with at most \(k\) crossings, no edge of \(G[V(I) \cup V(C)]\) is crossed and contracting \(I\) in such a drawing leads to a conforming drawing of \(G'\) with at most \(k\) crossings.
				\item Conversely, in any conforming drawing \(\mathcal{G}'\) of \(G'\) with at most \(k\) crossings, replacing \(v_I\) and its incident edges by an arbitrary planar drawing of \(G[V(I) \cup V(C)]\) which is homeomorphically distorted so that the drawing of \(C\) coincides with the drawing of \(C\) in \(\mathcal{G}'\), leads to a conforming drawing of \(G\) with at most \(k\) crossings. 
			\end{enumerate}
		\end{enumerate}
	\end{theorem}

	\subsection*{Characterising \ppd planarity}

	For the sake of completeness, we repeat here the original formulation of the main theorem of \cite{JelinekKR13}, and prove that it implies our wording used in Theorem~\ref{thm:PEG-relaxed-obstructions}.
	We shortly say that a \ppdg graph $\ca P=(G,\ca{H})$ is {\em planar} if $G$ admits a planar drawing which extends $\ca H$.

	We start with the skipped formal definition of an alternating chain from Theorem~\ref{thm:PEG-relaxed-obstructions}.
	For a cycle $C$ and vertices $x,y,x',y'\in V(C)$, we say that the pair $x,y$ alternates with the pair $x',y'$ on $C$ if they are distinct and their cyclic order on $C$ is $x,x',y,y'$ or $x,y',y,x'$.
	An {\em alternating chain} in a \ppdg graph $(G,\ca{H})$ is a \ppdg subgraph $(G_1,\ca{H}_1)$ of $(G,\ca{H})$ such that:
	(a) $\ca H_1$ is composed of a cycle $C$ and two isolated vertices~$s,t$, and
	(b) for some $a\geq2$ there exist $a$ paths $P_1,\ldots,P_a\subseteq G_1$ such that, for $i\in[a]$, $P_i$ has both ends on $C$ and is otherwise disjoint from $C$, and for $i\in[a-1]$ we have $P_i$ disjoint from $P_{i+1}$ and the ends of $P_i$ alternate with the ends of $P_{i+1}$ on~$C$.%
	\footnote{Note that with this condition we are more relaxed than \cite{JelinekKR13}; our definition of an alternating chain includes also situations which contain another obstruction from Figure~\ref{fig:PEG-obstructions} as well (as a PEG-minor), while \cite{JelinekKR13} are strict in excluding the other obstructions from the definition.
	For instance, obstructions number 16 and 4 from Figure~\ref{fig:PEG-obstructions} are included as alternating chains of $a=2,3$, respectively.
	There is more difference, e.g., three paths of the chain are allowed in our definition to pairwise alternate their ends on $C$, and this situation contains the obstruction $K_{3,3}$ even without predrawing.
	All these differences are irrelevant with respect~to~Theorem~\ref{thm:PEG-relaxed-obstructions}.}
	Furthermore, $s\in V(P_1)$, $t\in V(P_a)$, and $s$ and $t$ are predrawn in $\ca H_1$ in the same face of $C$ if $a$ is even, and in distinct faces of $C$ if $a$ is odd.

	We now define the {\em PEG-minor} relation on \ppdg graphs as introduced in \cite{JelinekKR13}.
	A \ppdg graph $\ca P=(G,\ca{H})$ contains a \ppdg graph $\ca P'=(G',\ca{H}')$ as a PEG-minor if a \ppdg graph isomorphic to $\ca P'$ can be obtained from $\ca P$ by a sequence composed of the following operations defined on~$\ca P$:
	\begin{enumerate}[I.]
		\item Remove an edge, or a vertex with all incident edges, both from $G$ and from~$\ca H$.
		\item Remove an edge, or a vertex with all incident edges, only from $\ca H$ while keeping them in~$G$ (in other words, the affected edge or vertex are no longer predrawn).
		\item Contract an edge of $H$ in both $G$ and $\ca H$, or contract an edge of $G$ which has at most one end in~$\ca H$.
		\item Assume that $f=uw\in E(G)\setminus E(H)$ has both ends in~$\ca H$. If $u$ and $w$ are in distinct components of $H$, and the degrees of $u$ and $w$ in $H$ are at most~$1$, then add $f$ to $\ca H$ and contract $f$ as in the previous point.
			\label{it:complicontr}
	\end{enumerate}
	Notice that applying any of these operations leaves a unique drawing of the \pdsg of $\ca P$ in the result (achieving this property is the reason for having rather complicated and restrictive definition of point \ref{it:complicontr}.).
	Furthermore, the operations clearly preserve \ppd planarity.

	We also need the more restrictive definition of an alternating chain from \cite{JelinekKR13} which we call here {\em reduced} for distinction.
	A {\em reduced alternating chain} is a \ppdg graph $(G_1,\ca{H}_1)$, consisting of $\ca H_1$ formed by a plane cycle $C$ and two isolated vertices~$s,t$,
	of the vertex set $V(G_1)=V(H_1)=V(C)\cup\{s,t\}$ and, for some $a\geq3$, of a set $F=E(G_1)\setminus E(C)$ of $a+2$ edges such that:
	\begin{itemize}
		\item The vertices $s$ and $t$ are predrawn in $\ca H_1$ in the same face of $C$ if $a$ is even, and in distinct faces of $C$ if $a$ is odd.
		\item The subgraph $(V(G_1),F)=G_1-E(C)$ has all vertices of degree $2$ except two vertices $u_1,v_1\in V(C)$ of degree~$1$.
		\item Edges of $F$ form $a$ paths where $P_1$ is formed by the two edges incident to $s$, $P_a$ is formed by the two edges incident to $t$, and $P_2,\ldots,P_{a-1}$ are formed by the remaining single edges of $F$ in some order.
			Moreover, $u_1$ is an end of $P_2$ and $v_1$ is an end of~$P_{a-1}$.
		\item For $i,i\in[a]$, the ends of $P_i$ and of $P_j$ alternate on $C$ if and only if~$|i-j|=1$.
	\end{itemize}

	\begin{theorem}[Jel\'{\i}nek, Kratochv\'{\i}l and Rutter \cite{JelinekKR13}]\label{thm:PEG-obstructions}
		A \ppdg graph $\ca P=(G,\ca{H})$ is planar if and only if $\ca P$ does not contain a reduced alternating chain or a member of the family $\cf L$ of the $24$ \ppdg graphs from Figure~\ref{fig:PEG-obstructions} as a PEG-minor.
	\end{theorem}

	\PEGourversion*
	\begin{proof}
		We show how our formulation of this theorem follows from original Theorem~\ref{thm:PEG-obstructions}.

		As for \eqref{it:noalter}, if there is an alternating chain in $\ca P$ then, directly, $\ca P$ cannot be planar.
		On the other hand, assume that $\ca P$ contains a reduced alternating chain as a PEG-minor.
		Then one can follow operations of a PEG-minor in the converse direction and routinely verify that there is an alternating chain in $\ca P$.
		However, notice that one cannot automatically claim that distinct single-edge paths from the definition of a reduced alternating chain give raise to internally disjoint paths with ends on $C$ -- this we claim only for the consecutive alternating pairs of them.

		Regarding \eqref{it:nosubdiv}, we have to properly define the family~$\cf K$.
		We say that a \ppdg graph $\ca P_1$ contains $\ca P_2$ as a {\em poor PEG-minor} if $\ca P_2$ is obtained as a PEG-minor above, but using only the removal operations from points I., II., and the contraction operation from point III. applied to edges having some end of degree at most $2$ in~$G$.
		Then, obviously, $\ca P_1$ contains a \ppdg subgraph isomorphic to a subdivision of~$\ca P_2$.
		Let $\cf K$ be a class of \ppdg graphs such that;
		\begin{itemize}
		\item every member of $\cf K$ contains some graph from the family $\cf L$ (cf.~Theorem~\ref{thm:PEG-obstructions}) as a PEG-minor,
		\item every \ppdg graph $\ca P$ which contains some graph from the family $\cf L$ as a PEG-minor also contains a member of $\cf K$ as a poor PEG-minor, and
		\item $\cf K$ is a minimal such class up to isomorphism of \ppdg graphs and the poor PEG-minor containment.
		\end{itemize}

		It immediately follows from this definition and Theorem~\ref{thm:PEG-obstructions} that if $\ca P$ contains a member of $\cf K$ as a poor PEG-minor, then $\ca P$ cannot be planar.
		Conversely, if $\ca P$ is not planar and $\ca P$ contains no alternating chain, then $\ca P$ contains a member of $\cf L$ as a PEG-minor (Theorem~\ref{thm:PEG-obstructions}), and so a member of $\cf K$ as a poor PEG-minor.
		Then a subdivision of a member of $\cf K$ is isomorphic to a \ppdg subgraph of $\ca P$.

		It remains to show that $\cf K$ is finite (we do not need it to be unique).
		Assume $\ca P_1\in\cf K$ which contains $\ca P_2\in\cf L$ as a PEG-minor.
		Then no removal operation is used when obtaining $\ca P_2$ from $\ca P_1$, by minimality of $\cf K$, and so no contraction of an edge with one end of degree $\leq2$ in point III. as well.
		Consequently, the underlying graph of $\ca P_1$ results using a bounded number of possible vertex splittings in the underlying graph of $\ca P_2$, and a bounded number of possible releasings of predrawn edges.
		This implies finiteness of $\cf K$ from the assumed finiteness of~$\cf L$.
	\end{proof}

	\section{Additions to Section~\ref{sec:phase1} -- Treewidth bound}

	The delicate and important concept of \(\mathcal{(H, I)}\)-flippability can be fully formally captured with the following elaborate construction.
	\begin{definition}
		\label{def:flip}
		Let \(D \supseteq C \cup H \cup I\) be a graph that admits a planar drawing whose restriction to \(H\) is equivalent to \(\mathcal{H}\).
		Let \(D'\) arise from \(D\) by adding a disjoint triangle \(T\) and three new vertices \(t\), \(i_1\) and \(i_2\), subdividing an edge \(e \in E(C)\) (if possible \(e \in E(C) \cap E(H)\)) with four new vertices \(v_1\) through \(v_4\), and introducing the edges \(tv_2\), \(tv_3\), \(v_2i_1\), \(v_3i_2\), and \(i_1n_1\) where \(n_1\) is a neighbour in \(I\) of the first vertex in \(N(I)\) when traversing \(C\) from \(v_2\) in direction away from \(v_3\), and \(i_2n_2\) where \(n_2\) is a neighbour in \(I\) of the first vertex in \(N(I)\) when traversing \(C\) from \(v_3\) in direction away from \(v_2\).
		We call the triangle subgraph on \(t\), \(v_2\), and \(v_3\), \(A\).
		
		Let \(\mathcal{H_T}\) arise from \(\mathcal{H}\) by adding an arbitrary planar drawing of \(T\) into the outer face of \(\mathcal{H}\) without enclosing any other vertices of \(D'\).
		We say a drawing of \(H \cup T \cup A \cup \{v_1v_2,v_2i_1,v_3v_4\}\) is of type 1 (type 2 respectively) if its restriction to \(T \cup A \cup \{v_1v_2,v_2i_1,v_3v_4\}\) is equivalent to the corresponding drawing in \Cref{fig:types}.
		Of course, if \(e \in E(H)\) we homeomorphically distort the drawing of the \(v_1v_2v_3v_4\)-path to be atop the drawing of \(e\) in \(\mathcal{H}\).
		Otherwise the depicted drawing is in the interiour of an arbitrary face of \(\mathcal{H}\).
		
		A cycle \(C\) in \(G\) is \emph{\(\mathcal{(H, I)}\)-flippable} in \(D\) if \(D'\) admits two planar drawings whose restrictions to \(H \cup T \cup A \cup \{v_1v_2,v_2i_1,v_3v_4\}\) are a type 1 and a type 2 drawing respectively.
		Otherwise, \(C\) is called \emph{\(H\)-unflippable} in \(D\).
	\end{definition}
	\begin{figure}
		\begin{center}\includegraphics[scale=0.8]{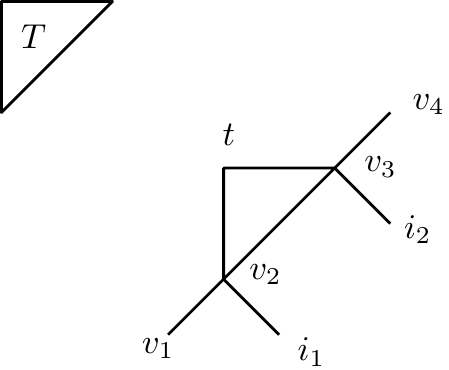} \hfil\vline\hfil \includegraphics[scale=0.8,page=2]{types}
		\end{center}
		\caption{Types for \Cref{def:flip}: type 1 left, type 2 right.\label{fig:types}}
	\end{figure}
	
	Note that, we can enumerate all (linearly many) possible faces for the placement of the drawing of \(A\) and for each face use \cite{AngeliniDFJVKPR15} to check in linear time whether a cycle is \(\mathcal{(H, I)}\)-flippable in a given graph.

	Essentially,  \(\mathcal{(H, I)}\)-flippability captures the fact that the orientation of \(C\) with respect to \(I\) in a planar drawing of \(D\) that is equivalent to \(\mathcal{H}\) on \(H\) is not determined by \(\mathcal{H}\).
	We give an intuitive explanation what we mean by this.
	For us to glue a drawing of \(I\) into a drawing of the \ppdg graph in which \(I\) was contracted using \(C\) as an interface, it is important that, if \(I\) and \(v_I\) are on the same side of \(C\) -- without loss of generality its inside -- in the respective drawings, both drawings of \(C\) have the same cyclic order on the drawing of \(C\).
	Otherwise issues as depicted in \Cref{fig:incompatible} can arise.
	
	For the same cycle \(C\) in two different graphs that is unflippable for two \(D_1\) and \(D_2\) and two different partial drawings \(\mathcal{H}_1\) and \(\mathcal{H}_2\) such that \(E(C) \cap E(\mathcal{H}_1) = E(C) \cap E(\mathcal{H}_2)\), we say \(C\) must have the \emph{same orientation} if for the same choice of \(v_1\), \(v_2\) and the drawing of \(T\) in the above definition, \(D_1\) and \(D_2\) only admit the same type of drawing.
	
	The definition of flippability uses the drawing of \(T\) and \(A\) to point to the outer face and fix the cyclic order on the drawing of \(C\) in type 1 and type 2 drawings respectively.
	The role of the edges \(v_2i_1\), \(i_1n_2\), \(v_3i_2\) and \(i_2n_2\), and the way in which \(v_2i_1\) and \(v_3i_2\) are fixed in the rotation scheme around \(v_2\) and \(v_3\) forces \(I\) to be drawn inside \(C\).
	
	See \Cref{fig:flipex} for a detailed example how these newly introduced concepts behave in the situation from \Cref{fig:incompatible}.
	\begin{figure}
		\begin{center}
			\includegraphics[scale=0.8,page=4]{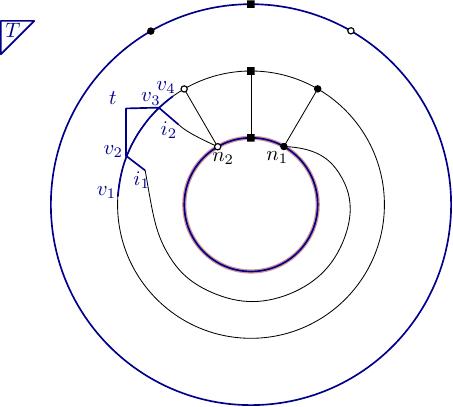}
			\hfill
			\includegraphics[scale=0.8,page=5]{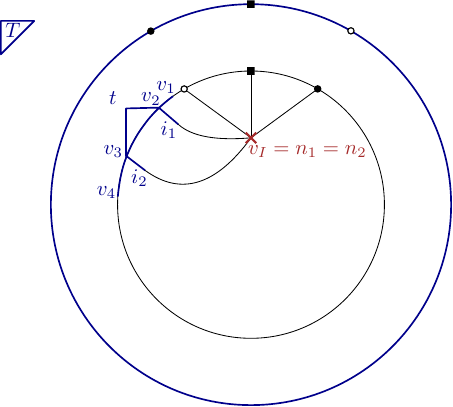}
		\end{center}
	\caption{A drawing of type 1 for the example from \Cref{fig:incompatible} (left). There is no drawing of type 2, which means the example is \((\mathcal{H},I)\)-unflippable.
	On the other hand, after contracting \(I\) to \(v_I\) there is a type 2 drawing (right), and contracting \(I\) in the left drawing immediately gives a type 1 drawing.
	This means the contraction results in \((\mathcal{H}',\{v_I\})\)-flippability.\label{fig:flipex}}
	\end{figure}

	We now show the main technical theorem of Phase~I formally (cf.\ \Cref{thm:reducegrah}).
	\begin{theorem}[Detailed version of \Cref{thm:reducegrah}]
		\label{thm:maintechnical}
		For all \(k \in \mathbb{N}\) there is some \(w \in \mathbb{N}\), such that given a \ppdg graph \((G, \mathcal{H})\) in which some edges are marked `uncrossable' in~\FPT-time parameterised by \(k\) we can
		\begin{enumerate}
			\item decide that there is no \conf drawing of \((G, \mathcal{H})\); or
			\item find a tree decomposition of \(G\) of width at most \(w\); or
			\item find a connected subgraph \(I \subseteq G\) and a cycle \(C\) in \(G - V(I)\) that contains the neighbourhood of \(I\) with the following properties:
			Consider the \ppdg graph \((G', \mathcal{H}')\) that arises from \(G\) by contracting \(I\) to the vertex \(v_I\) in \(G\) and contracting \(I \cap H\) in \(\mathcal{H}\), we mark all edges as `uncrossable' that are in \(C\), or incident to \(v_I\), or in \(E(G)\) and marked as `uncrossable'.
			\begin{enumerate}
				\item \label{property:peplanar} There is a \crconf{\(0\)} drawing of \((G[V(I) \cup V(C)] \cup H, \mathcal{H})\).
				\item In any \conf drawing of \((G,\mathcal{H})\), no edge of \(G[V(I) \cup V(C)]\) is crossed and contracting \(I\) in such a drawing leads to a \conf drawing of \((G', \mathcal{H}')\).
				\item If \(C\) is \(\mathcal{(H, I)}\)-flippable, in \(G[V(I) \cup V(C)] \cup H\), in any \conf drawing \(\mathcal{G}'\) of \((G', \mathcal{H}')\), replacing \(v_I\) and its incident edges by the restriction of an arbitrary \crconf{\(0\)} drawing of \((G[V(I) \cup V(C)] \cup H, \mathcal{H})\) to \(G[V(I) \cup V(C)]\), potentially after applying the mirroring from \Cref{def:flip} to obtain a drawing of \(G[V(I) \cup V(C)]\) in which \(C \cup H\) is drawn equivalently as in \(\mathcal{G}'\), which is homeomorphically distorted so that the drawing of \(C \cup H\) coincides with the drawing of \(C \cup H\) in \(\mathcal{G}'\) leads to a \conf drawing of \((G,\mathcal{H})\).
				\item If \(C\) is \(\mathcal{(H, I)}\)-unflippable in \(G[V(I) \cup V(C)] \cup H\) and \(C\) is \((\mathcal{H}',\{v_I\})\)-unflippable in \(G'[\{v_I\} \cup V(C)] \cup H'\), in any \conf drawing of \((G', \mathcal{H}')\), replacing \(v_I\) and its incident edges by the restriction of an arbitrary \crconf{\(0\)} drawing of \((G[V(I) \cup V(C)] \cup H, \mathcal{H})\) to \(G[V(I) \cup V(C)]\) which is homeomorphically distorted so that the drawing of \(C \cup H\) coincides with the drawing of \(C \cup H\) in \(\mathcal{G}'\) leads to a \conf drawing of \((G,\mathcal{H})\).
				\item Otherwise, \(C\) is \(\mathcal{(H, I)}\)-unflippable in \(G[V(I) \cup V(C)] \cup H\) and \((\mathcal{H}',\{v_I\})\)-flippable in \(G'[\{v_I\} \cup V(C)] \cup H'\).
				We can find \(u_1,u_2,u_3 \in N(I)\) such that replacing \(v_I\) by a triangle \(T_I\) and connecting all neighbours of \(v_I\) starting from \(u_3\) up to \(u_1\) to one corner, all neighbours of \(v_I\) starting from \(u_1\) up to \(u_2\) to the next corner, and all neighbours of \(v_I\) starting from \(u_2\) up to \(u_3\) to the last corner in a cyclic ordering of the neighbours of \(v_I\), and adding a specified drawing of the triangle to \(\mathcal{H}'\) results in the \((\mathcal{H}',T_I)\)-unflippability of \(C\) in \(G'[V(T_I) \cup V(C)] \cup H'\).
				Further, after this modification of \(G'\) and \(\mathcal{H}'\), in any \conf drawing of \((G', \mathcal{H}')\), one can replace \(T_I\) and its incident edges by the restriction of an arbitrary \crconf{\(0\)} drawing of \((G[V(I) \cup V(C)] \cup H, \mathcal{H})\) to \(G[V(I) \cup V(C)]\), if any exists, which is homeomorphically distorted so that the drawing of \(C \cup H\) coincides with the drawing of \(C \cup H\) in \(\mathcal{G}'\) leads to a \conf drawing of \((G,\mathcal{H})\).
			\end{enumerate}
		\end{enumerate}
	\end{theorem}
	\begin{proof}
		We start by applying result by Grohe~\cite{Grohe04} for \(k\) with \(G\) as input.
		In the case that the algorithm decides that the number of crossings in any drawing of \(G\) in which no `uncrossable' edge is crossed is more than \(k\), we can safely return that the same is true for any such drawing that is equivalent to \(\mathcal{H}\) on the \pdsg.
		Similarly if the algorithm returns a tree decomposition of width at most \(w\), we can return that tree decomposition.
		
		In the last case the algorithm finds a subgraph \(I \subseteq G\) and a cycle \(C\) in \(G\) as described in the preliminaries for Grohe's result for classical crossing number.
		We distinguish whether there is a \crconf{\(0\)} drawing of \((G[V(I) \cup V(C)] \cup H, \mathcal{H})\), or not.
		Recall that, as we assume \(\mathcal{H}\) to be planarised, edges marked as `uncrossable' are irrelevant in this context because no edge should be crossed.
		Hence deciding whether there is a \crconf{\(0\)} drawing of \((G[V(I) \cup V(C)] \cup H, \mathcal{H})\) is equivalent to deciding whether \(\crgpd(G[V(I) \cup V(C)] \cup H, \mathcal{H}) = 0\).
		This can be decided in linear time using the result by Angelini et al.~\cite{AngeliniDFJVKPR15}.
		
		\begin{case}
			There is no \crconf{\(0\)} drawing of \((G[V(I) \cup V(C)] \cup H, \mathcal{H})\).
		\end{case}
		In this case we claim that there is no \conf drawing of \((G, \mathcal{H})\).
		Assume for contradiction that there is such a drawing \(\mathcal{G}\).
		In particular, this drawing has at most \(k\) crossings and in it no `uncrossable' edge is crossed.
		Hence because of the choice of \(I\) and \(C\), no edge of \(G[V(I) \cup V(C)]\) is crossed in \(\mathcal{G}\).
		But as there are exactly \(\crd(\mathcal{H})\) crossings involving only edges of \(H\) in \(\mathcal{G}\), this means that the restriction of \(\mathcal{G}\) to \(G[V(I) \cup V(C)]\) is a \crconf{0} drawing of \((G[V(I) \cup V(C)] \cup H, \mathcal{H})\); a contradiction.
		
		\begin{case}
			There is a \crconf{\(0\)} drawing of \((G[V(I) \cup V(C)] \cup H, \mathcal{H})\).
		\end{case}
		In this case we show that \(I\) and \(C\) have the desired properties described under Point 3. of the theorem statement.
		
		Let \(G'\) arise from \(G\) by contracting \(I\) to \(v_I\), and let \(\mathcal{H}'\) arise from \(\mathcal{H}\) by contracting \(I\).
		
		Property a.\todo{better use label/ref and not fixed letters here} is exactly the case assumption.
		
		We now check b.
		Let \(\mathcal{G}\) be a \conf drawing of \((G, \mathcal{H})\).
		As \(\mathcal{G}\) has at most \(k\) crossings and in it no `uncrossable' edge is crossed, by the choice of \(I\) and \(C\), no edge of \(G[V(I) \cup V(C)]\) is crossed in \(\mathcal{G}\), and contracting \(I\) in \(\mathcal{G}\) leads to a drawing \(\mathcal{G}'\) of \(G'\) with at most \(k\)-crossings in which no `uncrossable' edge is crossed.
		It remains to show that \(\mathcal{G}'\) restricted to \(H'\) is equivalent to \(\mathcal{H}'\).
		Consider the homeomorphism from \(\mathbb{R}^2\) to itself that takes the restriction of \(\mathcal{G}\) to \(H\) onto \(\mathcal{H}\).
		Apart from \(v_I\) and its incident edges this also takes the restriction of \(\mathcal{G}'\) to \(H'\) onto \(\mathcal{H}'\).
		This means that if \(I \cap H = \emptyset\) we already have a homeomorphism as desired.
		
		Otherwise we additionally map the subset of \(\mathbb{R}^2\) corresponding to the face of \(\mathcal{G}' - \{v_I\}\) in which both of the contraction vertices \(v_I\) and \(v_{I \cap H}\) lie (this is well-defined because \(I\) is drawn planar in \(\mathcal{G}\) and hence also \(I \cap H\) is drawn planar in \(\mathcal{H}\)) homeomorphically to itself while mapping the contraction point of \(I\) in \(\mathcal{G}'\) to the contraction point of \(I \cap H\) in \(\mathcal{H}'\).
		In this way all endpoints of edges of \(H\) incident to \(v_I\) are taken to the appropriate endpoints of edges incident to \(v_{I \cap H}\).
		Now it is straightforward to see that the edge drawings themselves can also be mapped onto each other by a homeomorphism of \(\mathbb{R}^2\) to itself.
		The combination of these homeomorphisms describes a homeomorphism from \(\mathbb{R}^2\) to itself that takes the restriction of \(\mathcal{G}'\) to \(H'\) onto \(\mathcal{H}'\).
		
		Now we check c. through e.
		
		Informally speaking, if we find an \(\mathcal{(H, I)}\)-flippable cycle \(C\) we will essentially be able to flip any planar drawing of the contracted subgraph to appropriately match the interface in a drawing of \(G\) after contraction.
		
		If we find a cycle that is \(\mathcal{(H, I)}\)-unflippable and the cycle remains unflippable after the contraction of the subgraph is performed, any planar drawing of the contracted subgraph automatically matches the interface in a drawing of \(G\) after contraction.
		
		The last case is that the cycle we find is \(\mathcal{(H, I)}\)-unflippable but seems to be flippable after the contraction of the subgraph is performed.
		In this case the orientation of the cycle is fixed in any planar drawing of the contracted subgraph, but both orientations of the cycle are possible after the contraction of the subgraph.
		We must therefore appropriately force the orientation of \(C\) in the drawing after preforming the contraction to match the one which is in fact forced before the contraction.
		We will do this by extending \(\mathcal{H}\) carefully.
		
		In all three of c. through e. it will be important to observe that the case distinction on the flippability of \(C\) ensures that in each case the drawings of \((G, \mathcal{H})\) defined from drawings of \((G' \mathcal{H}')\) are actually well-defined (recall that we define contractions in subgraphs of drawings only for connected and uncrossed subgraphs).
		
		\noindent \textbf{c.} \hspace{.5em} Consider the case that \(C\) is \(\mathcal{(H, I)}\)-flippable in \(G[V(I) \cup V(C)] \cup H\) and let \(\mathcal{G}'\) be an arbitrary \conf drawing of \((G', \mathcal{H}')\).
		Also consider an arbitrary .
		According to \Cref{def:flip}, we find a \crconf{\(0\)} drawing \(\mathcal{I}\) of \((G[V(I) \cup V(C)] \cup H, \mathcal{H})\) in which \(C\) has the same orientation with respect to \(\mathcal{H}'\) as in \(\mathcal{G}'\) and which is equivalent to \(\mathcal{H}\) on \(H \cap I\).
		This means there is a homeomorphism from the 2-dimensional sphere to itself that takes the drawing restricted to the connected component of \(G'[V(C) \cup V(H')] - \{v_I\}\) containing \(C\) to \(\mathcal{G}'\) restricted to \((C \cup H) - \{v_I\}\).
		Let \(\mathcal{I}\) be the restriction of the resulting homeomorphic drawing to \(G[V(I) \cup V(C)]\).
		Replacing \(v_I\) and its incident edges by \(\mathcal{I}\) then results in a drawing of \(G\) which coincides with \(\mathcal{G}'\) on \(H' - \{v_I\}\) and is equivalent to \(\mathcal{H}\) on \(H \cap I\).
		Overall this implies the drawing is equivalent to \(\mathcal{H}\) on \(H\).
		Moreover \(\mathcal{G}'\) has at most \(k\) crossings none of which involve `uncrossable' edges, and \(\mathcal{I}\) is planar.
		Because of the choice of \(I\) and \(C\),
		defining \(\mathcal{G}\) from \(\mathcal{G}'\) by replacing \(v_I\) and its incident edges by \(\mathcal{I}\) leads to a drawing of \(G\) with at most \(k\) crossings in which no `uncrossable' edge is crossed, i.e.\ this shows the drawing is \conf.
		
		\noindent \textbf{d.} \hspace{.5em} Next, consider the case that \(C\) is \(\mathcal{(H, I)}\)-unflippable in \(G[V(I) \cup V(C)] \cup H\) and \((\mathcal{H}',\{v_I\})\)-unflippable in \(G'[\{v_I\} \cup V(C)] \cup H'\).
		Let \(\mathcal{G}'\) be an arbitrary \conf drawing of \((G', \mathcal{H}')\).
		Also consider an arbitrary \crconf{\(0\)} drawing \(\mathcal{G_I}\) of \((G[V(I) \cup V(C)] \cup H, \mathcal{H})\).
		We claim that these two drawings are equivalent on \((C \cup H') - \{v_I\}\).
		Both drawings are equivalent on \(C\) and \(H' - \{v_I\}\) separately.
		For them to be equivalent on both \(C\) and \(H' - \{v_I\}\) simultaneously, it is sufficient that \(C\) has the same orientation in both drawings.
		By \Cref{def:flip}, as \(C\) is \(\mathcal{(H, I)}\)-unflippable in \(G[V(I) \cup V(C)] \cup H\) the orientation of \(C\) is uniquely determined, given a fixed drawing of \(H\).
		Contracting \(I\) in \(\mathcal{G_I}\) leads to a planar drawing of \(G'[\{v_I\} \cup V(C)] \cup H'\) which is equivalent to \(\mathcal{H}'\) on \(H'\) and equivalent to \(\mathcal{G_I}\) on \((C \cup H') - \{v_I\}\).
		By \Cref{def:flip}, as \(C\) is also \((\mathcal{H}',\{v_I\})\)-unflippable in \(G'[\{v_I\} \cup V(C)] \cup H'\) and by the choice of `uncrossable' edges in \((G',\mathcal{H})\), the restriction of \(\mathcal{G}'\) to \(G'[\{v_I\} \cup V(C)] \cup H'\) is also planar and the orientation of \(C\) in \(\mathcal{G}'\) is the same in \(\mathcal{G_I}\).
		
		Let \(\mathcal{I}\) be the restriction of \(\mathcal{G_I}\) to \(G[V(I) \cup V(C)]\).
		Again, replacing \(v_I\) and its incident edges by \(\mathcal{I}\) then results in a drawing of \(G\) which coincides with \(\mathcal{G}'\) on \(H' - \{v_I\}\) and is equivalent to \(\mathcal{H}\) on \(H \cap I\).
		Overall this implies the drawing is equivalent to \(\mathcal{H}\) on \(H\).
		Moreover \(\mathcal{G}'\) has at most \(k\) crossings none of which involve `uncrossable' edges, and \(\mathcal{I}\) is planar.
		Because of the choice of \(I\) and \(C\),
		defining \(\mathcal{G}\) from \(\mathcal{G}'\) by replacing \(v_I\) and its incident edges by \(\mathcal{I}\) leads to a drawing of \(G\) with at most \(k\) crossings in which no `uncrossable' edge is crossed, i.e.\ this shows the drawing is \conf.
		
		\noindent \textbf{e.} \hspace{.5em} Finally, consider the case that \(C\) is \(\mathcal{(H, I)}\)-unflippable in \(G[V(I) \cup V(C)] \cup H\) and \((\mathcal{H}',\{v_I\})\)-flippable in \(G'[\{v_I\} \cup V(C)] \cup H'\).
		We try all possible \(u_1, \dotsc, u_3 \in N(I)\) and then let \(\mathcal{H}'\) be defined as in the theorem statement and \(H'\) be the underlying graph that is drawn by \(\mathcal{H}'\) till we arrive at a choice in which \(C\) is \((\mathcal{H}',T_I)\)-unflippable.
		Such a graph must exist as otherwise \(C\) is \(\mathcal{(H, I)}\)-flippable in \(G[V(I) \cup V(C)] \cup H\), contradicting our case assumption (for instance, from a drawing of the type that exists one can use \(n_1\), \(n_2\) and a neighbour of \(I\) incident to an edge that can no longer be added planely if considering the other type of drawing).
		By connecting \(T_I\) according to the fixed orientation of \(C\) with respect to \(I\) that follows from the \(\mathcal{(H, I)}\)-unflippability of \(C\) in \(G[V(I) \cup V(C)] \cup H\), we can also fix the orientation of \(C\) with respect to \(T_I\).
		Now we can argue the remaining claims in this case analogously to d.
	\end{proof}
	Recalling our running example from \Cref{fig:incompatible}, we have already determined (\Cref{fig:framingex}) that it leads to Case~e.\ of \Cref{thm:maintechnical}.
	\Cref{fig:triangle} depicts the construction in this case for our concrete example.
	\begin{figure}
		\begin{center}
			\includegraphics[scale=0.6,page=6]{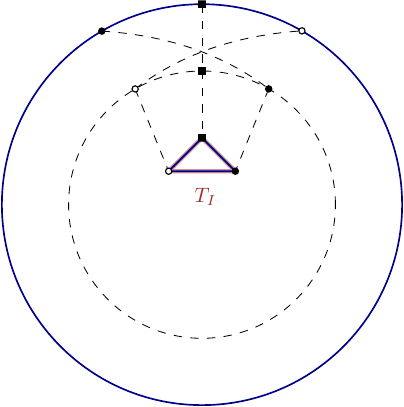}
		\end{center}
		\caption{Replacing of \(v_I\) by \(T_I\) for the example from \Cref{fig:framingex} to force the orientation that follows from the \((\mathcal{H},I)\)-unflippability of \(C\).
			Now, in contrast to the situation in \Cref{fig:incompatible}, we can correctly determine that no drawing extending \(\mathcal{H}\) with \(0\) crossings exists.\label{fig:triangle}}
	\end{figure}
	
	We can iteratively apply \Cref{thm:maintechnical} to reduce our instance to a graph of small treewidth;
	in each iteration that leads to Point 3. we decrease the size of the graph by at least three vertices and, given a solution for the smaller graph, can distinguish the flippability case in polynomial time (this is because the definition of flippability can trivially be checked in quadratic time for any given cycle and partially drawn graph) and use the result by Angelini et al.~\cite{AngeliniDFJVKPR15} to find a drawing to replace the contracted vertex in linear time.
	If at any point we run into Point 1. or are not able to find an appropriate drawing for the contracted subgraph, we can safely decide that we have a no-instance.

	\twbound*
	\begin{proof}
		Let \(g \in \mathbb{N}\) such that a grid of dimensions \(g \times g\) occurs as a minor of \(F\).
		
		We can simply argue that the inner \((g-2) \times (g-2)\)-subgrid cannot contain any edge which corresponds to a path in \(F\) consisting only of edges of a framing cycle:
		Assume for a contradiction that there are two vertices \(v,w\) from the same framing cycle \(C\) around \(c \in V(\mathcal{H}^\times)\) such that a path \(P\) in \(C\) between \(v\) and \(w\) corresponds to an edge of the subgrid.
		After removing \(P\) from \(F\), the neighbourhood of \(v,w\) consists of at most five vertices; namely~\(c\), at most one neighbour each on \(C - V(P)\), and at most one vertex each for which the edge between \(c\) and that vertex leads to \(v\) and \(w\) respectively.
		However, two vertices that are incident on the inner subgrid of a grid always must have a neighbourhood of size at least six.
		A contradiction.
		
		This means the \((g - 2) \times (g - 2)\)-subgrid remains even if we contract every framing cycle with the vertex of \(\mathcal{H}^\times\) around which it was introduced.
		After this contraction we are left with \((G - E(H)) \cup \mathcal{H}^\times\) in which some edges are trippled and further triples of possible connector edges are added.
		
		We now also argue that within the inner \((g-2) \times (g-2)\)-subgrid we can assume that there are fewer than \(4^{k + 1}\) edges (where \(k = \crgpd(G, \mathcal{H})\)\,) corresponding to a path in \(F\) that contains connector edges:
		Firstly, on paths in which we can instead of a connector edge use a parallel edge from \(G\), we do so.
		Hence we only focus on connector edges where no edge of \(G\) could be used to planarly connect two previously disconnected components of \(F\) during Step 1 of Definition~\ref{def:framing}.
		By the construction, connector edges can be viewed as edges between connected components of \(\mathcal{H}^\times\), and contracting each of these components results in a tree \(T\) with the connector edges as the edge set.
		As every vertex of the subgrid has degree \(4\) in the entire grid, in case there are more than \(4^{k + 1}\) such edges, we find at least two \(e\) and \(f\) are at distance at least \(k + 1\) in \(T\).
		By the construction of \(F\) this means \(e\) and \(f\) connect connected components of \(\mathcal{H}^\times\) that are separated by at least \(k + 1\) closed curves in \(\mathcal{H}^\times\).
		There must be at least one path in \(F\) that does not use any connector edge or edge in \(G\) which is parallel to a connector edge, corresponding to a path in the grid between the endpoints of \(e\) and~\(f\).
		That means this path is entirely in \(G\) and connects two faces of \(\mathcal{H}^\times\) that are separated by at least \(k + 1\) closed curves in \(\mathcal{H}^\times\), and by the Jordan curve theorem~\cite{Veblen05} it necessarily has crossings with each of these curves.
		This would imply that \(\crgpd(G, \mathcal{H}) > k\).
		
		Removing all edges from the \((g-2) \times (g-2)\)-subgrid which require connector edges leaves a subgrid of size at least \((g - 4^{k + 1} - 2) \times (g - 4^{k + 1} - 2)\) intact, which is a minor of \(G^o=(G - E(H)) \cup \mathcal{H}^\times\).
		By the grid minor theorem~\cite{RobertsonST94} this implies that $$g - 4^{k + 1} - 2 \in \mathcal{O}(\sqrt{\operatorname{tw}(G^o)/\log(\operatorname{tw}(G^o))}).$$
		On the other hand applying the grid minor theorem for \(F\) implies \(g \in \mathcal{O}(\sqrt{\operatorname{tw}(F)/\log(\operatorname{tw}(F))})\).
		This concludes our proof.
	\end{proof}

	\section{Additions to Section~\ref{sec:MSOenc} -- MSO\(_2\)-encoding}

	Following Definition~\ref{def:framing}, we now define an extended framing which generalizes the former in a way that will be necessary when looking for a \ppdg subgraph in the form of a subdivision of a fixed \ppdg graph.
	Essentially, the problem addressed by this extended definition is the one of having a disconnected \pdsg in the subgraph, which is embedded such that the \pdss\ components are connected together indirectly, through other components not existing in the subgraph.

	Let $\ca P=(G,\ca H)$ be a \ppdg graph where $\ca H$ is planar. We call a face of $\ca H$ {\em rich} if it is incident with more than one connected component of~$\ca H$.
	If $\ca H$ is connected, then an {\em extended framing} of $\ca P$ is just an ordinary framing of~$\ca P$.
	If, otherwise, $\ca H$ has $c>1$ connected components, then an {\em extended framing} of $\ca P$ is a graph $\bar H'\cup G'$, where a \ppdg graph $\ca P'=(G',\ca H')$ is defined next and $\bar H'$ is a framing of just $(H',\ca H')$.
	The \ppdg graph $\ca P'$, called an {\em extended framing base}, is an arbitrary graph obtained from $\ca P=(G,\ca H)$ by adding at most $2c-2$ (possibly zero) new vertices, which are called the {\em connector vertices} and are placed into rich faces of $\ca H$ or subdivide edges bounding rich faces.

	\begin{lemma}\label{lem:finframings}
		Let $\ca P_1=(G_1,\ca{H}_1)$ be a \ppdg graph. Then the set of all distinct (up to isomorphism) extended framings of $\ca P_1$ is finite.
	\end{lemma}
	\begin{proof}
		Definition~\ref{def:framing} constructs a framing of a \ppdg graph in a deterministic way, except the arbitrary choice of connector edges when the \pdsg is not connected.
		There are finitely many choices for these connector edges in a finite graph.
		Furthermore, in the case of \ppdg graph $\ca P_1$ with disconnected $\ca H_1$, there are finitely many non-isomorphic choices of edge subdivisions and additions of isolated vertices to faces with at most $2c-2$ new vertices in total.
	\end{proof}

	Before getting to the full proof of Lemma~\ref{lem:restrtopol}, we repeat the definition of a framing topological minor in a more detailed setting.
	Considering framings $\bar G_1$ of $(G_1,\ca{H}_1)$ and $\bar G_2$ of $(G_2,\ca{H}_2)$, and let $N_1$ and $N_2$ in this order denote the sets of connector edges used in the framings $\bar G_1$ and~$\bar G_2$.
	For every edge $f\in V(H_i)\cup N_i$, $i=1,2$, let $t(f)$ denote the framing triplet of $f$ in $\bar G_i$, and for $v\in V(H_i)$ let $s(v)$ denote the framing cycle around $v$ in~$\bar G_i$.

	\begin{definition}\label{def:framingtopol}
		We say that $\bar G_1$ is a {\em framing topological minor} of $\bar G_2$ if there is a topological-minor embedding of $\bar G_1$ into $\bar G_2$ (i.e., a subgraph-isomorphism mapping of a subdivision of $\bar G_1$ into $\bar G_2$) which additionally satisfies
		\begin{itemize}
			\item every edge of $G_1$ (respectively., of $H_1$) is mapped -- within the topological-minor embedding of $\bar G_1$ into $\bar G_2$ -- into a path of $G_2$ (respectively., of $H_2$), and every edge of $N_1$ is ``virtually'' (via the framing triplets) mapped into a path of $H_2+N_2$,
			\item whenever a predrawn vertex $v\in V(H_1)$ is mapped into $v'\in V(H_2)$, the framing cycle $s(v)$ in $\bar G_1$ is embedded in the corresponding framing cycle $s(v')$ in~$\bar G_2$,
			\item whenever a predrawn edge $f\in E(H_1)$ is mapped into a path $P_f\subseteq H_2$ where $P_f=(v_0,v_1,\ldots,v_m)$,
			the framing triplet $t(f)$ of $f$ in $\bar G_1$ is embedded (as three internally-disjoint paths) in the subgraph $\bigcup_{i=0}^{m-1}t(v_iv_{i+1})\cup\bigcup_{i=1}^{m-1}s(v_i)$ of $\bar G_2$, such that one of the three paths $t(f)$ maps to passes through all vertices of~$P_f$, and
			\item whenever a connector edge $f\in N_1$ is mapped into a path $P_f\subseteq H_2+N_2$, the same condition as in the previous point holds for~$t(f)$.
		\end{itemize}
	\end{definition}

	\encodeintopol*
	\begin{proof}
	In the forward direction, let us assume that $\ca P_0$ is a subdivision of $\ca P_1$ such that $\ca P_0\subseteq\ca P_2$.
	Let $N_2$ denote the set of connector edges which have been used when constructing the framing $\bar G_2$ of $\ca P_2$, cf.~Definition~\ref{def:framing}.
	Then $N_2$ is formed by edges with ends in~$H_2$ and the graph $\widetilde H_2:=H_2+N_2$ is connected.
	Let $\ca P_0=(G_0,\ca H_0)$ and note that $H_0$ is a subgraph of~$H_2$.
	Choose an (arbitrary) graph $H_0^+$ such that $H_0^+$ is connected, $H_0\subseteq H_0^+\subseteq\widetilde H_2$, and $H_0^+$ is inclusion-minimal of these properties.

	If $H_0$ is connected, then we are in the trivial case of $H_0^+=H_0$, but we need to investigate the other case in which $H_0$ has $c>1$ connected components.
	Observe that after contracting in $H_0^+$ each connected component of $H_0$ into a single vertex, we get by minimality a tree~$T$.
	Every leaf of $T$ is a former component of $H_0$, and the internal vertices may be either former components of $H_0$ or vertices of $X:=V(T)\cap\big(V(H_2)\setminus V(H_0)\big)$.
	In particular, $T$ has at most $c$ leaves. An easy calculation then shows that $X$ contains at most $c-2$ vertices of degree~$>2$,
	and together at most $2c-2$ vertices of $T$ are not degree-$2$ members of~$X$.
	Consequently (and recalling that $H_0$ is a subdivision of $H_1$), $H_0^+$ is a subdivision of some extended framing base of~$\ca P_1$.
	Let $\bar G_1$ be the extended framing of $\ca P_1$ which comes from this base.

	Our claim is that $\bar G_1$ is a framing topological minor of $\bar G_2$.
	This is immediate regarding the original (non-frame) edges of $G_1$ in $\bar G_1$ since $\ca P_0\subseteq\ca P_2$ is a subdivision of $\ca P_1$.
	Regarding the frame edges of~$G_1$, we proceed along the definition of a framing topological minor given above.
	For the edges $f$ of $H_1$ and the connector edges of $\bar G_1$, we have got the corresponding framing triplets in $\bar G_2$ along the path $P_f\subseteq H_2+N_2$ that $f$ maps into.
	It remains to easily observe that the union of these triplets and their framing cycles indeed contains three internally disjoint paths between the ends of $P_f$.
	The framing triplet of $f$ then maps into them in the right order.

	\medskip
	We similarly proceed in the backward direction of the proof.
	Assume that an extended framing $\bar G_1$ of $\ca P_1$ is a framing topological minor of~$\bar G_2$.
	Then, in particular, a subdivision of $G_1$ is a subgraph of $G_2$ such that the corresponding subdivision of $H_1$ belongs to $H_2$.
	The frame edges of $\bar G_1$ form a $3$-connected planar graph $F_1$ (on the vertex set $V(H_1)$), and $F_1$ is by the definition mapped in $\bar G_2$ into a subgraph $F_2$ isomorphic to a subdivision of $F_1$.
	Both $F_1$ and $F_2$ hence have unique planar drawings, and since $\ca H_1$ and $\ca H_2$ are both planar, we conclude that the drawing $\ca H_1$ is equivalent to the corresponding restriction of $\ca H_2$.
	Then the subdivision of $\ca P_1$ is a \ppdg subgraph of $\ca P_2$.
	\end{proof}

	\MSOforobstruction*
	\begin{proof}
		Let $\cf F$ be the set of all distinct extended framings of $\ca P_1$ (which is finite by Lemma~\ref{lem:finframings}).
		Using Lemma~\ref{lem:restrtopol}, we write the formula $\sigma\equiv\bigvee_{\bar G_1\in\cf F}\sigma[\bar G_1]$ where $\bar G_2\models\sigma[\bar G_1]$ expresses that $\bar G_1$ is a restricted topological minor of~$\bar G_2$.
		Note that the graph $\bar G_1$ is a constant in the formula $\sigma[\bar G_1]$, i.e., $\bar G_1$ is hardcoded in $\sigma[\bar G_1]$ which is then evaluated only on $\bar G_2$.
		Expressing $\sigma[\bar G_1]$ in MSO$_2$ is then a routine task -- we simply do the following:
		\begin{itemize}
			\item Guess (using a new existential quantifier for each one) for all vertices $v$ of $G_1$ their distinct images in $V(G_2)$
			and for all edges $f$ of $G_1$ the edge sets $E_f\subseteq E(G_2)$, where $E_f$ is intended to contain the edges of the path that $f$ maps to according to Definition~\ref{def:framingtopol}.
			\item Guess for every connector edge $f$ of $\bar G_1$ the edge set $E_f\subseteq E(\bar G_2)\setminus E(G_2)$ of frame edges restricted only to framing triplets in $\bar G_2$, where $E_f$ is intended to represent one path of the framing triplet of~$f$ in~$\bar G_2$.
			\item Verify consistency of the previous edge sets; namely, that the sets are pairwise disjoint, that for every edge $f$ the set $E_f$ indeed induces a path in $\bar G_2$ with ends corresponding to the ends of $f$,
			and that no vertex is incident at the same time to an edge of $E_f$ and of $E_{f'}$, $f'\not=f$, unless it is an end of both paths.
			\item Guess for every predrawn or connector edge $f\in E(H_1)\cup N_1$, where $N_1$ is the set of the connector edges of~$\bar G_1$, the set $E^+_f\subseteq E(\bar G_2)\setminus E(G_2)$,
			and check the following property with respect to $f$ and the previous set $E_f$:
			if $E_f$ induces a path $P_f=(v_0,v_1,\ldots,v_m)$ in $\bar G_2$ then, precisely as in Definition~\ref{def:framingtopol}, we have that $E^+_f=\bigcup_{i=0}^{m-1}t(v_iv_{i+1})\cup\bigcup_{i=1}^{m-1}s(v_i)$.
			\item For every predrawn vertex $v\in V(H_1)$, let $f_1,\ldots,f_d$ be the edges of $H_1+N_1$ incident to~$v$, and $v_1\in V(H_2)$ be the vertex $v$ maps to in $\bar G_2$.
			The framing triplets of $f_1,\ldots,f_d$ determine a cyclic order on the framing cycle $s(v)$ of~$v$.
			Verify that the sets $E^+_{f_1},\ldots,E^+_{f_d}$ are hitting the framing cycle $s(v_1)$ of $v_1$ in $\bar G_2$ in the same order.
			Furthermore, for $i\in[d]$ and $f_i=vw$, verify the existence of three disjoint paths in $\bar G_2$ from the framing cycle $s(v_1)$ to $s(w_1)$ contained in $E^+_{f_i}$
			(this way we ensure that the considered cyclic orders around $v_1$ and around $w_1$ in $\bar G_2$ are not flipped).
	\end{itemize}
	Validity of this formula $\sigma[\bar G_1]$ follows directly from Definition~\ref{def:framingtopol}.
	\end{proof}

	\encodealternating*
	\begin{proof}
		We encode (in MSO$_2$) the existence of an alternating chain $\ca P_0=(G_0,\ca{H}_0)\subseteq^?\,\ca P_2$ as follows.
		Recall that $\ca H_0$ consists of a plane cycle $C$ and vertices $s,t$ (positioned within $C$ according to the parity of the chain), and we check $(H_0,\ca H_0)$ as a \ppdg subgraph of $\ca P_2$ directly as in Lemma~\ref{lem:MSOforobstruction} on~$\bar G_2$.
		Then we have to check existence of the $a$ paths $P_1,\ldots,P_a\subseteq G_2$ as in the definition of an alternating chain.
		Note that this is not possible directly (i.e., checking path by path) since $a$ is not bounded and we do not have a separate variable for each of the paths.
		Instead, we are going to use the following standard trick.

		Imagine we have (or ``guess'' with existential quantifiers in $\tau$) the unions $A=P_1\cup P_3\cup\dots$ and $B=P_2\cup P_4\cup\dots$ of the odd- and even-indexed paths of the assumed chain (which are not disjoint in general).
		We say that two of our paths $P_i$ and $P_j$ {\em alternate} if, as in the definition of an alternating chain, $V(P_i)\cap V(P_j)=\emptyset$ and the ends of $P_i$ and $P_j$ alternate on~$C$.
		Since the length $a$ of our assumed alternating chain is not really important -- only the parity of $a$ matters, we just check the existence of a chain between $P_1$ and $P_a$ of the required parity as follows:
		For every subgraph $D\subseteq A\cup B$ such that $P_1\subseteq D$ and $P_a\not\subseteq D$, we check that there exist $i,j\in[a]$ such that one of the following is true; $P_i\subseteq A\cap D$ (respectively., $P_i\subseteq B\cap D$), $P_j\not\subseteq D$ and $P_j\subseteq B$ (respectively., $P_j\subseteq A$),
		and $P_i$ and $P_j$ alternate.
		If we started from an actual alternating chain, this check clearly succeeds. On the other hand, if this check succeeds, we can find an alternating chain whose paths come from $A\cup B$ and with the required parity (which is implicitly checked with $P_a\subseteq A$ or $P_a\subseteq B$).

		To summarise the proof, we construct $\tau$ as follows:
		\begin{itemize}
			\item The formula $\tau$ is a disjunction of two similar subformulas (for $a$ odd and $a$ even), which are described together in the next points.
			\item Let $\ca P_1=(H_1,\ca H_1)$ be a \ppdg graph in which $H_1$ is formed by $K_3$ and two extra vertices $s,t$, and $\ca H_1$ specifies a planar drawing with both $s,t$ in the same face ($a$ even) or in distinct faces ($a$ odd) of the triangle.
			We use the formula $\sigma$ of Lemma~\ref{lem:MSOforobstruction} to find (i.e., to guess using existential quantifiers) a framing of some subdivision $(H_0,\ca H_0)$~of~$\ca P_1$ in the given framing $\bar G_2$ of $\ca P_2$.
			This way we obtain the cycle $C$ of $\ca P_0$ in addition to~$s,t$.
			\par For the rest of the construction, we refer only to the vertices and edges of $\bar G_2$ which are from $G_2$ (i.e., the frame of $\ca P_2$ will no longer be relevant).
			\item We construct a subformula $\pi(F)$ which, for a given edge-set variable $F$, checks that $F$ is the edge set of a path which has both ends on $C$ and is otherwise disjoint from~$C$.
			This involves a standard MSO connectivity check and testing the degrees of the vertices incident to~$F$.
			\item Using existential quantifiers, we guess two edge sets $F_A$ and $F_B$ disjoint from $E(C)$ (imagine these as $F_A=E(A)$ and $F_B=E(B)$ in the above sketch),
			and such that there is no subset $X\subseteq F_A\cap F_B$ for which $\pi(X)$ would hold.%
			\footnote{Notice that we cannot simply demand that $F_A\cap F_B=\emptyset$, as our alternating chain is a splitting of a reduced alternating chain used in \cite{JelinekKR13} as a PEG-minor.}
			\item We moreover guess sets $F_1\subseteq F_A$ and $F_2\subseteq F_B$ if $a$ is even (respectively., $F_2\subseteq F_A$ if $a$ is odd).
			We check that $\pi(F_1)$ and $\pi(F_2)$ hold and that $s$ is incident to $F_1$ and $t$ is incident to~$F_2$
			(in the above sketch this corresponds to $F_1=E(P_1)$ and $F_2=E(P_a)$).
			\item For every edge set $F_D$ (using a universal quantifier), check the following.
			If $F_D\subseteq F_A\cup F_B$ such that $F_1\subseteq F_D$ and $F_2\not\subseteq F_D$, then there must exist $F_3,F_4\subseteq F_A\cup F_B$ such that
			\begin{itemize}
				\item $\pi(F_3)$ and $\pi(F_4)$ hold, $F_3\subseteq F_D$ and $F_4\not\subseteq F_D$, and no vertex is incident both to an edge of $F_3$ and to an edge of~$F_4$,
				\item the two vertices of $C$ incident to edges of $F_3$ alternate on $C$ with the two vertices incident to edges of $F_4$, and
				\item either $F_3\subseteq A$ and $F_4\subseteq B$, or $F_4\subseteq A$ and $F_3\subseteq B$.
			\end{itemize}
		\end{itemize}

		To finish the proof, in one direction, we can routinely verify that if $\ca P_0\subseteq\ca P_2$ is an alternating chain in the given \ppdg graph $\ca P_2$, then~$\bar G_2\models\tau$.
		In the other direction, if $\bar G_2\models\tau$, then we obtain the \ppdg subgraph $(H_0,\ca H_0)$ and the edge sets $F_A$ and~$F_B$.
		To derive that this implies the existence of an alternating chain $\ca P_0\supseteq(H_0,\ca H_0)$ in $\ca P_2$, observe that $\tau$ actually verifies that the following bipartite graph is connected:
		The vertex set of this auxiliary graph is formed by the sets $X\subseteq F_A\cup F_B$ satisfying $\pi(X)$, and $X,X'$ are adjacent if and only if the paths induced by them alternate on~$C$, and $X\subseteq F_A$, $X'\subseteq F_B$ or vice versa.
	\end{proof}

	\MSOphaseII*
	We remark that it could be tempting to generalise this theorem in a way that some edges of $G$ would carry a numeric upper bound (specially, bound~$0$ for an `uncrossable' edges) on the allowed number of crossings on them, but there is a subtle catch.
	This generalisation would be possible, within the current proof method, only if $\ca H$ was planar or, at least, the edges carrying the bounds would not be crossed in $\ca H$ or would be from the rest of~$G$.
	Otherwise, while planarising $\ca H$ at the beginning, the bounds on planarised edges would be lost in $\ca H^\times$ and could not be in general recovered by other means in the formula $\psi_k$.
	See also the discussion in Section~\ref{sec:per-edge}.
	\begin{proof}		
		Here we continue with additional technical details for the proof sketch presented in Section~\ref{sec:PPDCR}.
		Recall that we want to express, by a formula $\psi'_k$, that a ``guessed planarisation'' of the input graph $\ca P=(G,\ca H)$ indeed admits a planar drawing extending predrawn~$\ca H$.

		In order to keep this proof self-contained, we avoid using full-featured logical transductions for the task of ``guessing the crossings'' of the desired planarisation, and instead use what is called a simple interpretation.
		To make this task even simpler, we have subdivided each edge of $\ca P$ which is not marked as `uncrossable' by $k$ new vertices, called {\em auxiliary vertices} of this subdivision $\ca P_0=(G_0,\ca{H}_0)$ of~$\ca P$.

		Now we can capture the planarisation of assumed $k$-crossing conforming drawing as giving the $k$ pairs of auxiliary vertices to be identified in the planarisation.
		Formally, let $\vect r'=(r_i':i\in[k])$ and $\vect r''=(r_i'':i\in[k])$ be two $k$-tuples of auxiliary vertices of $\ca P_0$ which are pairwise identified as $\vect r'=\vect r''$, producing the graph $\ca P_0[\vect r'=\vect r'']$.
		With $\bar G_0$ a framing of $\ca P_0$, we aim to get $\psi'_k$ such that $[\bar G_0,\vect r',\vect r'']\models\psi'_k$ if and only if $\ca P_0[\vect r'=\vect r'']$ has a planar drawing extending~$\ca H_0$.

		Recall the obstruction class $\cf K$ from Theorem~\ref{thm:PEG-relaxed-obstructions}.
		Let $\sigma_{\cf K}$ be a formula which is the conjunction of the formulas $\sigma$ from Lemma~\ref{lem:MSOforobstruction} over all~$\ca P_1\in\cf K$, let $\tau$ be the formula from Lemma~\ref{lem:MSOalternating}, and~$\varrho\equiv\sigma_{\cf K}\wedge\tau$.
		Then, by Theorem~\ref{thm:PEG-relaxed-obstructions}, for a framing $\bar G_1$ of any \ppdg graph $\ca P_1$ we have that $\bar G_1\models\varrho$ if and only if $\ca P_1$ extends to a planar drawing.
		It remains to modify the formula $\varrho$ into $\psi'_k$ which, informally, ``speaks about $\bar G_0[\vect r'=\vect r'']$ inside $\ca P_0$'',
		i.e., $\bar G_0[\vect r'=\vect r'']\models\varrho$ $\iff$ $[\bar G_0,\vect r',\vect r'']\models\psi'_k$.
		This is just a straightforward syntactical translation of $\varrho$ which re-defines the incidence predicate between a vertex and an edge and the equality relation in $\varrho$ to respect the identification~$\vect r'=\vect r''$.
		The re-defined incidence predicate for $\psi'_k$ can be written as
		$$ \prebox{inc}'(v,e) \>\equiv\> \prebox{inc}(v,e)\vee \bigvee\nolimits_{i\in[k]}\big(v=r'_i\wedge\prebox{inc}(r''_i,e)\big)\vee \bigvee\nolimits_{i\in[k]}\big(v=r''_i\wedge\prebox{inc}(r'_i,e)\big) $$
		and the new equality predicate $='$ as
		$$ (v='w) \>\equiv\> (v=w) \vee \bigvee\nolimits_{i\in[k]}\big(v=r'_i\wedge w=r''_i\big)\vee \bigvee\nolimits_{i\in[k]}\big(v=r''_i\wedge w=r'_i\big) .$$
		The proof is finished.
	\end{proof}

	We can now finish the main result:
	\thmmain*
	\begin{proof}		
		Given a \ppdg graph $(G_1,\ca H_1)$ and an integer $k>0$, the task is to decide whether $\crgpd(G_1,\ca H_1)\leq k$.
		Recall that we would prefer to have $\ca H_1$ planar, and so let $G$ be the graph obtained from $G_1$ by planarising its subgraph $\ca H_1$ into~$\ca H_1^\times=:\ca H$.
		We obviously have $\crgpd(G_1,\ca H_1)=\crgpd(G,\ca H)$.
		
		Then, using repeatedly Theorem~\ref{thm:maintechnical}, we either arrive at the conclusion that that $\crgpd(G,\ca H)>k$, or we reduce the input to an equivalent instance $(G',\ca H')$ with the same inquired solution value~$k$ and with $\operatorname{tw}(G')\leq w$.
		Note that each intermediate iteration of this process decreases $|V(G)|$, and so there are only linearly many steps here.
		Moreover, using Lemma~\ref{lem:tw}, we get that the tree-width of any framing $\bar G'$ of $(G',\ca H')$ is bounded in terms of~$w$, and so in terms of our parameter~$k$.
		We can hence efficiently (in \FPT\ time) decide whether $\crgpd(G,\ca H)=\crgpd(G',\ca H')\leq k$ using Courcelle's theorem applied with the formula $\psi_k$ from Lemma~\ref{lem:phaseII} to a framing $\bar G'$ of $(G',\ca H')$.
		
		{We remark that dependence on the instance size of this procedure is given by \(\mathcal{O}(|V(G)|^3)\).
		This is because all components of the construction in \Cref{thm:maintechnical} have complexity \(\mathcal{O}(|V(G)|^2)\); specifically the procedures required from Grohe~\cite{Grohe04} and Angelini et al.~\cite{AngeliniDFJVKPR15} can be executed in linear time, 
		while the checking the flippability case when a smaller graph needs to be constructed can straightforwardly be implemented in quadratic time.
		The number of necessary iterations is clearly in \(\mathcal{O}(|V(G)|)\) as each iteration is the last iteration or reduces the size of the graph.
		This dominates the linear dependence of the final application of Courcelle's theorem on \(|V(G)|\).
		}
	\end{proof}

	\section{Additions to Section~\ref{sec:per-edge} -- Tracking original edges in the planarisation}
	We give the details to show \Cref{thm:Simple}.
	Recall that our approach for \textsc{\PPDCRc{\(c\)}} does not immediately work in this setting as to ensure simplicity we need to be able to identify edges of \(\mathcal{H}^\times\) that correspond to the same edge in \(H\).

	Edges incident to vertices inserted for crossing points in the planarisation of \(\mathcal{H}\) can easily be identified according to the original edge of \(H\) they arise from, as has previously been done in \cite{GanianHKPV21} in the following way.
	This means every edge in \(e \in E(\mathcal{H}^\times)\) receives up to two labels, and we speak about the label of \(e\) with respect to one of its endpoints.
	Specifically, we assign label \(0\) to both edges of \(\mathcal{H}^\times\) with respect to a single crossing vertex \(c\) that correspond to one arbitrary edge of \(H\) involved in the crossing at \(c\).
	Then we assign label \(1\) to the remaining two edges of \(\mathcal{H}^\times\) that are incident to \(c\), which automatically correspond to the other edge of \(H\) involved in the crossing at \(c\).
	In this way a path in the planarisation of \(\mathcal{H}\) corresponds to part of a single edge in \(H\) if and only if internally it consists of a sequence of an edge labelled with \(l \in \{0,1\}\) with respect to a certain crossing vertex \(c\), followed by that vertex \(c\) followed by the unique other edge incident to \(c\) labelled with \(l\).
	For a path in \(\mathcal{H}^\times\), if this is the case we say the path is an \emph{\Hep}.
	Whether a path is an \Hep is something we can easily express in MSO\(_2\) with the addition of only a constant-size edge label set to a framing of our bounded treewidth graph.
	
	However, contracting a subgraph of \(G\) leads to a major obstacle for using this method, as the vertex created by contracting the subgraph potentially has many incident edges some of which are unlabelled, others are labelled by \(0\) and others are labelled by \(1\).
	Now it is no longer possible to uniquely and correctly associate edges with the same label to each other.
	This is a difficulty that is not easy to overcome and a direction we leave open in this work.

	In the following we denote \(t = |E(G) \setminus E(H)|\).
	
	The following lemma will allow us to bridge the information for identifying relevant edges in \(\mathcal{H}^\times\) that correspond to the same edge in \(H\) for which the labelling cannot be preserved during the iterative contraction of subgraphs in Phase~I.
	\begin{lemma}[adapted from {\cite[Section~3]{GanianHKPV21}}]
		\label{lem:holes}
		Consider a \ppdg graph \((G, \mathcal{H})\) and a closed curve \(D\) together with a subgraph \(I_D\) of \((G - E(H)) \cup \mathcal{H}^\times\) that is separated in \((G - E(H)) \cup \mathcal{H}^\times\) by removing all edges of \(\mathcal{H}^\times\), such that in any \conf drawing of \(G\) in which each edge in \(E(G) \setminus E(H)\) has at most \(c\) crossings (we treat \(\mathcal{H}\) and \(\mathcal{H}^\times\) as drawn atop each other canonically), no edge of \(((G - E(H)) \cup \mathcal{H}^\times)[V(D) \cup V(I_D)]\) is crossed.
		In \FPT-time parameterised by \(c\) and \(t\) we can compute an embedded \(\{0,1\}\)-edge-labelled subgraph \(\tilde{I}_D \supseteq D\) of \(I_D \cup D\) with the following properties.
		\begin{itemize}
			\item The number of edges in \(E(\tilde{I_D}) \setminus E(D)\) is bounded by a function in \(c\) and \(t\).
			\item In every (possibly non-simple) \conf drawing of \(G\) in which each edge in \(E(G) \setminus E(H)\) has at most \(c\) crossings, if the planarisation of an edge \(e \in E(H)\) contains at least two edges in \(E(\mathcal{H}^\times)\) each with an endpoint in \(\mathcal{H}^\times\) on \(D\) and \(e\) is crossed twice by a single edge in \(E(G) \setminus E(H)\) then these crossings occur on two parts of the planarisation of \(e\) in \(\mathcal{H}^\times\) that are connected by an \Hep in \(\tilde{I}_D\).
		\end{itemize}
	\end{lemma}
	
	We can use this subgraph for appropriately chosen cycles based on Phase~I of our algorithm, and then slightly extend our MSO\(_2\)-expression from \Cref{lem:phaseII} to show:
	\Simple*
	\begin{proof}
		Obviously we can check for the given \ppdg \((G, \mathcal{H})\) whether \(\mathcal{H}\) is simple and decide that we have a no-instance if this is not the case.
		
		Now we label all edges of \(\mathcal{H}\) with \(\{0,1\}\) to associate the pairs of edges incident to a vertex introduced for a crossing in \(\mathcal{H}\) that belong to the same edge of \(H\).
		Next we perform the Phase I computations explained in \Cref{sec:phase1} for the given \ppdg \((G, \mathcal{H})\).
		After this we can assume to arrive at a \ppdg graph \((\tilde{G}, \tilde{\mathcal{H}})\) in which some edges are marked as `uncrossable' together with a set of possibly non-disjoint cycles \(\mathfrak{C}\) consisting only of `uncrossable edges' and special vertices in \(\tilde{G}\).
		Each of the special vertices is associated to a contracted planar graph of \(G - E(H) \cup \mathcal{H}^\times\).
		We also can assume to have a tree decomposition of \(\tilde{G}\) of width bounded in \(c \cdot t\).
		
		The edges of \(H\) which we cannot trace using the simple \(\{0,1\}\)-labelling after the contraction performed at the end of the iteration, are among those whose planarisations contain at least two edges in \(E(\mathcal{H}^\times)\) each with an endpoint in \(\mathcal{H}^\times\) on some \(C \in \mathfrak{C}\).
		We call these edges \emph{problematic}.
		However, we cannot simply apply \Cref{lem:holes} to the cycles in \(\mathfrak{C}\), as it in general is not a cycle in \(\tilde{\mathcal{H}}^\times\) and we hence do not know enough about its behaviour in all targeted drawings to appropriately define curves to separate \(I\).
		
		Instead let us focus on a single cycle \(C\) with associated subgraph \(I\) that is contracted to \(v_I\).
		Note that the simple \(\{0,1\}\)-labelling also works to trace an edge \(e \in E(\mathcal{H})\) if no vertex of the planarisation of \(e\) in \(E(\mathcal{H}^\times)\) is in \(I\).
		As \(I\) is connected and any vertex \(v \in V(I)\) of the planarisation of a problematic edge is a vertex introduced into the planarisation of \(\mathcal{H}\) because of a crossing and hence \(v \notin V(G)\), we know that all vertices of the planarisation of problematic edges in \(V(I)\) are connected to each other by edges of \(I \cap \mathcal{H}^\times\), and in fact only such edges that are not incident to vertices of \(G\).
		We can trace the outer boundary of the according connected component of \(I \cap \mathcal{H}^\times\) with a simple closed curve at \(\varepsilon\)-distance, and apply \Cref{lem:holes} to this curve.
		
		We replace \(v_I\) (both in \(\tilde{G}\) and \(\tilde{\mathcal{H}}\) by the resulting embedded graph \(I_C\), reconnecting original neighbours of vertices on \(C\) which is by definition a subgraph of \(I_C\).
		Then we contract all edges of \(C\) whose endpoints do not have any neighbours in \(I_C\).
		
		Doing this for every \(C \in \mathfrak{C}\) that intersects some problematic edge, means we replace each \(v_C\) by a graph of size bounded in \(c\) and \(t\).
		(This means replacing \(v_C\) by \(V(I_C)\) in every bag of the given tree decomposition only increases its width by a term bounded in \(c\) and \(t\).)
		And, most importantly, this replacement introduce \Hep~s to correctly identify edges in \(\mathfrak{H}^\times\) belonging to the same problematic edges of \(H\).
		It is also easy to see that \conf drawings of \((\tilde{G},\tilde{H})\) before and after the described replacement of the contraction vertices are equivalent.
		
		By \Cref{lem:tw}, and the fact that \(c \cdot t \geq \crgpd(G, \mathcal{H})\), the treewidth of \(\tilde{G} - E(\tilde{\mathcal{H}})\) after this modification is bounded in \(c\) + \(t\).
		
		Finally we can use Courcelle theorem on the conjunction of our MSO\(_2\)-encoding from \Cref{lem:phaseII} and a MSO\(_2\)-formula that excludes the existence of a \Hep between the vertices identified with any auxiliary vertex of the same edge in \(E(G) \setminus E(H)\), where we simply use a unique label for each edge in \(E(G) \setminus E(H)\) (this only introduces \(t\) new labels) to identify its auxiliary vertices.
		
		Correctness follows from \Cref{thm:maintechnical}, \Cref{lem:phaseII} and \Cref{lem:holes}.
	\end{proof}
	
	We remark that, while so far there are no results in literature on this, it is also natural to ask whether we can also determine the \ppdcr if we restrict the number of crossings with edges from \(E(G) \setminus E(H)\) for each edge in \(E(H)\) to be at most \(c\).
	This introduces the same difficulties as requiring simplicity of the completed drawing of \(G\).

	\section{Additions to Section~\ref{sec:conclu} -- Crossing-criticality}

	\baddualdiam*
	\begin{proof}
		We will deal in this proof using the {\em weighted crossing number}; that is, we consider a weight function $w:E(G)\to\mathbb N$ on the edges of our graph, such that a crossing between two edges $f$ and $f'$ contributes $w(f)\cdot w(f')$ towards the crossing number~$\crg(G)$.
		(We will later argue how to ``trade'' the weights for parallel edges in a rigorous way.)

		Consider the (ordinary) graph $G_1$ on $8$ vertices depicted on the left of Figure~\ref{fig:baddualdiam}.
		The weights of its edges are as follows; $w(u_1v_1)=w(u_3v_3)=c\geq3$, $w(u_4u_1)=w(v_4v_1)=w(u_iu_{i+1})=w(v_iv_{i+1})=2c+3$ for $i=1,2,3$, and the remaining weights are~$1$.
		Let $C_1\subseteq G_1$ be the cycle $(u_1,u_2,u_3,u_4)$ and $C_2\subseteq G_1$ the cycle $(v_1,v_2,v_3,v_4)$.
		We focus on drawings of the subgraph $G_0=C_1\cup C_2+u_2v_2+u_4v_4\subseteq G_1$.
		If a drawing of $G_0$ has at most $2c+2$ crossings, then it is planar since $c^2>2c+2$.
		Hence there are only two such drawings -- the {\em straight} one in which the cycles $C_1$ and $C_2$ agree in orientation, and the {\em flipped} one in which $C_2$ is drawn opposite to~$C_1$.

		Denote by $\ca D_1$ the planar drawing of $G_1$ (with straight subdrawing of $G_0$), and by $\ca D_2$ the drawing of $G_1$ (with flipped subdrawing of $G_0$) with $\crd(\ca D_2)=2c+2$ crossings shown in Figure~\ref{fig:baddualdiam} (left-bottom) alongside the drawing~$\ca D_1$.
		We claim:
		\begin{description}
		\item[(*)] The only drawings of $G_1$ with at most $2c+2$ crossings, up to equivalence, are those with the straight subdrawing of $G_0$ (such as $\ca D_1$), and the aforementioned drawing $\ca D_2$ and a symmetric picture to $\ca D_2$ (Figure~\ref{fig:baddualdiam} with $u_1v_1$ drawn above and $u_3v_3$ below) -- both with exactly $2c+2$ crossings.
		\end{description}
		Knowing that the subdrawing of $G_0$ must be planar and flipped, proving (*) is a matter of a straightforward analysis of the edges of $E(G_1)\setminus E(G_0)$.

	\begin{figure}[t]
	\begin{center}
	\begin{tikzpicture}[scale=0.7]
		\node at (-3,-2) {$\ca D_1$};
		\node at (-3,-8) {$\ca D_2$};
		\tikzstyle{every path}=[color=black]
		\draw[very thick](0,0) circle(1);
		\draw[very thick](0,0) circle(2);
		\draw[thick] (0,1) -- (0,2);
		\node at (0,1.5) [label=right:$c$] {};
		\node at (1.5,1.5) [label=right:$2c+3$] {};
		\draw[thick] (0,-1) -- (0,-2);
		\tikzstyle{every node}=[draw, shape=circle, minimum size=2pt,inner sep=1.4pt, fill=black]
		\node at (-2,0) (u1) [label=left:$u_1$] {}; \node at (-1,0) (v1) [label=right:$v_1$] {};
		\node at (0,-2) (u2) [label=below:$u_2$] {}; \node at (0,-1) (v2) [label=above:$v_2$] {};
		\node at (2,0) (u3) [label=right:$u_3$] {}; \node at (1,0) (v3) [label=left:$v_3$] {};
		\node at (0,2) (u4) [label=above:$u_4$] {}; \node at (0,1) (v4) [label=below:$v_4$] {};
		\tikzstyle{every path}=[color=green!50!black]
		\draw (u1) -- (v1);
		\draw (u3) -- (v3);
		\draw (u1) to[bend left=36] (v4);
		\draw (v4) to[bend left=36] (u3);
		\draw (v3) to[bend left=36] (u2);
		\draw (u2) to[bend left=36] (v1);
	\begin{scope}[shift={(0,-6)}]
		\tikzstyle{every path}=[color=black]
		\draw[very thick](0,0) circle(1);
		\draw[very thick](0,0) circle(2);
		\draw[thick] (0,1) -- (0,2);
		\draw[thick] (0,-1) -- (0,-2);
		\tikzstyle{every node}=[draw, shape=circle, minimum size=2pt,inner sep=1.4pt, fill=black]
		\node at (-2,0) (u1) [label=left:$u_1$] {}; \node at (-1,0) (v3) [label=right:$v_3$] {};
		\node at (0,-2) (u2) [label=below:$u_2$] {}; \node at (0,-1) (v2) [label=above:$v_2$] {};
		\node at (2,0) (u3) [label=right:$u_3$] {}; \node at (1,0) (v1) [label=left:$v_1$] {};
		\node at (0,2) (u4) [label=above:$u_4$] {}; \node at (0,1) (v4) [label=below:$v_4$] {};
		\tikzstyle{every path}=[color=green!50!black]
		\draw (u1) .. controls (-0.5,-2.7) and (1.3,-1.2) .. (v1);
		\draw (u3) .. controls (0.5,2.7) and (-1.3,1.2) .. (v3);
		\draw (u1) to[bend left=36] (v4);
		\draw (v4) to[bend left=27] (u3);
		\draw (v3) to[bend right=36] (u2);
		\draw (u2) to[bend right=45] (v1);
	\end{scope}
	\end{tikzpicture}
	\qquad
	\begin{tikzpicture}[scale=0.15]
		\tikzstyle{every path}=[color=white!30!black]
		\path[use as bounding box] (-20,-26) rectangle (33,26);
		\foreach \x in {10,11,12,13,16,17,18,19,20}
			\draw[very thick](0,0) circle(\x);
		\draw[very thick, dotted](0,0) circle(14.5);
		\draw[thick, dotted] (0,13) -- (0,16); 	\draw[thick, dotted] (0,-13) -- (0,-16);
		\draw[thick] (0,10) -- (0,13);  \draw[thick] (0,16) -- (0,20);
		\draw[thick] (0,-10) -- (0,-13);  \draw[thick] (0,-16) -- (0,-20);
		\tikzstyle{every node}=[draw, shape=circle, minimum size=2pt,inner sep=1.4pt, fill=black]
		\tikzstyle{every path}=[color=black]
		\node at (2,-5) (t12) {}; \node at (2,5) (t13) {}; \node at (-5,0) (t11) {};
		\node at (25,-5) (t22) {}; \node at (25,5) (t23) {}; \node at (32,0) (t21) {};
		\node at (-10,0) (v1) {}; \node at (0,-10) (v2) {}; \node at (0,10) (v4) {};
		\node at (20,0) (m1) {}; \node at (0,-20) (m2) {}; \node at (0,20) (m4) {};
		\draw[very thick] (t11) -- (v1) (t12) -- (v2) (t13) -- (v4);
		\draw[very thick] (t21) to[bend right=25] (m1) (t22) to[out=250,in=300] (m2) (t23) to[out=110,in=60] (m4);
		\tikzstyle{every path}=[color=blue]
		\draw[thick] (t11) -- (t12) -- (t13) -- (t11);
		\draw[thick] (t21) -- (t22) -- (t23) -- (t21);
		\node[draw=none,fill=none] at (0,0) {$T_1$};
		\node[draw=none,fill=none] at (27,0) {$T_2$};
	\end{tikzpicture}
	\end{center}
		\caption{The ``critical'' construction in the proof of Proposition~\ref{pro:baddualdiam}. Left: the graph $G_1$ with its two example drawings $\ca D_1$ and $\ca D_2$.
			Right: the \ppdg graph $\ca P'$ obtained by stacking $2m$ copies of $G_1$ with identification (only the ``thick'' edges are shown) and connecting predrawn triangles $T_1$ and $T_2$.}
		\label{fig:baddualdiam}
	\end{figure}
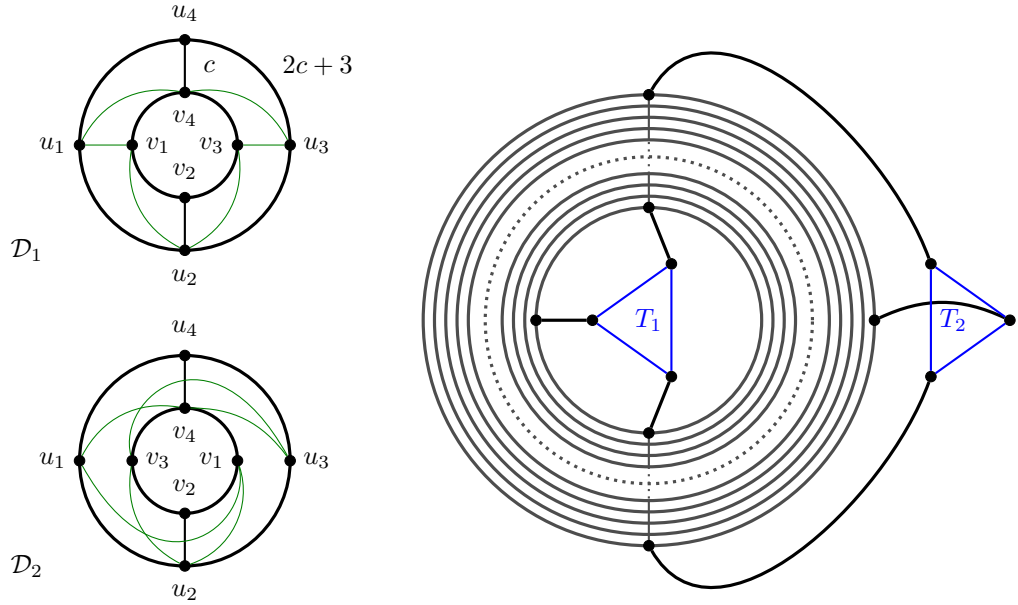

		Furthermore, from the pictures we easily derive the following ``criticality'' property:
		\begin{description}
		\item[(**)] If a graph $G_2$ is obtained from $G_1$ by deleting any one of the six edges of $E(G_1)\setminus E(G_0)$, then there is a drawing of $G_2$ with flipped subdrawing of $G_0$ and at most $2c+1$ crossings.
		\end{description}

		Now, for $i=1,\ldots,m'$ where $w'=2m$, we make a copy $G_1^i$ of the graph $G_1$, and construct a graph $G''$ from the union $G_1^1\cup\dots\cup G_1^{m'}$ by identifying the cycle $C_1^j=(u_1^j,u_2^j,u_3^j,u_4^j)$ with the cycle $C_2^{j+1}=(v_1^{j+1},v_2^{j+1},v_3^{j+1},v_4^{j+1})$ in order, for~$j=1,\ldots,m'-1$.
		We add two graph triangles $T_1$ and $T_2$, such that the vertices of $T_1$ are adjacent in order to $v_1^1,v_2^1,v_4^1$ of $G''$ by edges of weight $2c+3$, and similarly the vertices of $T_2$ are adjacent in order to $u_3^{m'},u_2^{m'},u_4^{m'}$ of $G''$ again by edges of weight $2c+3$.
		Let $G'$ be the resulting graph (Figure~\ref{fig:baddualdiam} right).

		We define $H=T_1\cup T_2$ and the drawing $\ca H$ of $H$ in which $T_1$ and $T_2$ are in opposite orientation (against the drawing in Figure~\ref{fig:baddualdiam}).
		Consider the \ppdg graph $\ca P'=(G',\ca H)$ and $k=2c+2$.
		We claim that $\crgpd(\ca P')\geq2c+2=k$.
		Indeed, the edges incident to $T_1$ or $T_2$ cannot be crossed in a drawing with $\leq2c+2$ crossings (unlike in Figure~\ref{fig:baddualdiam}), and hence by (*) exactly one of the graphs $G_1^j\subseteq G'$, $j\in[m']$, must be drawn with flipped subdrawing of $G_0^j$
		(and the remaining copies $G_1^{i}$ planarly). The claim follows.
		On the other hand, after deleting any edge of $E(G_1^i)\setminus E(G_0^i)$, $i\in[m']$, the crossing number drops below~$k$ by (**).

		\medskip
		In the last step we get rid of weighted edges in $\ca P'$ as follows. If $f\in E(G')$ is of weight $w>1$, we replace $f$ with a bunch of $w$ parallel edges.
		Let $\ca P^o=(G^o,\ca H)$ be the resulting unweighted \ppdg graph. We claim that $\crgpd(\ca P^o)\geq k$ as well.
		This is a standard argument in this area; considering an optimal drawing $\ca D^o$ of $\ca P^o$, we iteratively redraw edges of each parallel bunch tightly along the edge of the least number of crossings in $\ca D^o$ in this bunch.
		Consequently, $\crgpd(\ca P^o)=\crgpd(\ca P')$.

		Finally, $\ca P^o$ may not be $k$-crossing critical, but there exists a \ppdg subgraph $\ca P=(G,\ca H)\subseteq\ca P^o$ which is inclusion-minimal with the property $\crgpd(\ca P)\geq k$.
		Moreover, by~(**), $G$ contains all cycles $(u_1^j,v_1^j,u_2^j,v_3^j,u_3^j,v_4^j)$ of $G_1^j$ for~$j\in[m']$, which are drawn nested in an optimal drawing of~$\ca P$.
		Thus, $\ca P=(G,\ca H)$ satisfies all claimed properties.
	\end{proof}

\fi

\end{document}